\documentclass[acmsmall,authorversion,nonacm]{acmart}

\newif\ifshort

\usepackage{tablefootnote}
\usepackage[inline]{enumitem}

\usepackage{multirow}
\usepackage{xspace}
\usepackage{adjustbox}

\usepackage{graphicx}
\usepackage{caption}
\usepackage{subcaption}
\usepackage{lscape}

\newtheorem*{remark}{Remark}
\usepackage{amsthm}
\usepackage{amsmath}
\usepackage{bussproofs} 
    \EnableBpAbbreviations
\usepackage{mathpartir} 
\usepackage{mathtools}
\usepackage{stmaryrd}
\newcommand{\R}{\mathbb{R}}

\newcommand{\NNR}{\mathbb{R}^{\geq 0}}

\newcommand{\F}{\mathbb{F}}

\newcommand{\rnderr}{\ensuremath{\epsilon}}

\newcommand{\rnd}{\textbf{rnd }}
\newcommand{\ret}{\textbf{ret }}
\newcommand{\tlet}{\textbf{let }}
\newcommand{\tin}{\textbf{in }}

\newcommand{\letbind}{\textbf{let-bind}}
\newcommand{\inl}{\textbf{inl }}
\newcommand{\inr}{\textbf{inr }}
\newcommand{\op}{\textbf{op}}

\usepackage{listings}
\usepackage{lstlang}
\usepackage{xcolor}
\usepackage[T1]{fontenc}
\usepackage[scaled]{beramono}
\newcommand{\unit}{\textbf{unit}}
\newcommand{\num}{\textbf{num}}
\newcommand{\tensor}{\otimes}
\newcommand{\tand}{\ \& \ }
\newcommand{\lin}{\multimap}
\newcommand{\bang}[1]{{!_{#1}}}

\newcommand{\Met}{\mathbf{Met}}
\newcommand{\Set}{\mathbf{Set}}
\newcommand{\denot}[1]{\llbracket {#1} \rrbracket}
\newcommand{\pdenot}[1]{\llparenthesis {#1} \rrparenthesis}

\newcommand{\Lang}{$\Lambda_\num$\xspace}

\usepackage{quiver}

\usepackage{cleveref}

\begin{document}

\title{Numerical Fuzz: A Type System for Rounding Error Analysis}

\author{Ariel E. Kellison}
\orcid{0000-0003-3177-7958}
\affiliation{%
  \institution{Cornell University}
  \city{Ithaca}
  \country{USA}
}
\email{ak2485@cornell.edu}

\author{Justin Hsu}
\orcid{0000-0002-8953-7060}
\affiliation{%
  \institution{Cornell University}
  \city{Ithaca}
  \country{USA}
}
\email{justin@cs.cornell.edu}

\begin{abstract}
  Algorithms operating on real numbers are implemented as floating-point
  computations in practice, but floating-point operations introduce
  \emph{roundoff errors} that can degrade the accuracy of the result. We propose
  \Lang, a functional programming language with a type system that can express
  quantitative bounds on roundoff error. Our type system combines a sensitivity
  analysis, enforced through a linear typing discipline, with a novel graded
  monad to track the accumulation of roundoff errors. We prove that our type
  system is sound by relating the denotational semantics of our language
  to the exact and floating-point operational semantics.

  To demonstrate our system, we instantiate \Lang with error metrics proposed in
  the numerical analysis literature and we show how to incorporate rounding
  operations that faithfully model aspects of the IEEE 754 floating-point
  standard.  To show that \Lang can be a useful tool for automated error
  analysis, we develop a prototype implementation for \Lang that infers error
  bounds that are competitive with existing tools, while running significantly
  faster and scaling to larger programs. Finally, we consider semantic
  extensions of our graded monad to bound error under more complex rounding
  behaviors, such as non-deterministic and randomized rounding.
\end{abstract}

\ccsdesc[500]{Mathematics of computing~Numerical analysis}
\ccsdesc[500]{General and reference~Verification}
\ccsdesc[300]{Software and its engineering~Functional languages}

\keywords{Floating point, Roundoff error, Linear type systems}

\maketitle

\section{Introduction}\label{sec:introduction}

Floating-point numbers serve as discrete, finite approximations of continuous
real numbers. Since computation on floating-point numbers is designed to
approximate a computation on ideal real numbers, a key goal is reducing the
\emph{roundoff error}: the difference between the floating-point and the ideal
results. To address this challenge, researchers in numerical analysis have
developed techniques to measure, analyze, and ultimately reduce the
approximation error in floating-point computations.

\paragraph*{Prior work: formal methods for numerical software.}
While numerical error analysis provides a well-established set of tools for
bounding roundoff error, it requires manual effort and tedious calculation.  To
automate this process, researchers have developed verification methods based on
\emph{abstract interpretation} and \emph{optimization}. In the first approach,
the analysis approximates floating-point numbers with roundoff error by
\emph{intervals} of real numbers, which are propagated through the computation.
In the second approach, the analysis approximates the floating-point program by
a more well-behaved function (e.g., a polynomial), and then uses global optimization to find the maximum error between the approximation and
the ideal computation over all possible realizations of the rounding error.

While these tools are effective, they share some drawbacks. 
First, there is the issue of limited scalability: analyses that rely on global optimization are not compositional, and analyses based on interval arithmetic are compositional, but when the intervals are composed naively, the analysis produces bounds that are too large to be useful in practice. 
Second, existing tools largely focus on
small simple expressions, such as straight-line programs, and it is unclear how to
extend existing methods to more full-featured programming languages.

\paragraph*{Analyzing floating-point error: basics and challenges.}
To get a glimpse of the challenges in analyzing roundoff error, we begin with
some basics. At a high level, floating-point numbers are a finite subset of the
continuous real numbers. Arithmetic operations (e.g., addition, multiplication)
have floating-point counterparts, which are specified following a common
principle: the output of a floating-point operation applied to some arguments is
the result of the \emph{ideal} operation, followed by \emph{rounding} to a
representable floating-point number. In symbols:
\[
  \widetilde{op}(x, y) \triangleq \widetilde{op(x, y)}
\]
where the tilde denotes the approximate operation and rounded value,
respectively.

Though simple to state, this principle leads to several challenges in analyzing
floating point error. First, many properties of exact arithmetic do not hold for
floating-point arithmetic. For example, floating-point addition is not
associative: $(x \tilde{+} y) \tilde{+} z \neq x \tilde{+} (y \tilde{+} z)$.
Second, the floating point error accumulates in complex ways through the
computation. For instance, the floating point error of a computation cannot be
directly estimated from just the \emph{number} of rounding operations---the
details of the specific computation are important, since some operations may
amplify error, while other operations may reduce error. 


\paragraph*{Our work: a type system for error analysis.}
To address these challenges, we propose \Lang, a novel type system for error
analysis. Our approach is inspired by the specification of floating-point
operations as an exact (ideal) operation, followed by a rounding step. The key
idea is to separate the error analysis into two distinct components: a
\emph{sensitivity} analysis, which describes how errors propagate through the
computation in the absence of rounding, and a \emph{rounding} analysis, which
tracks how errors accumulate due to rounding the results of operations.

Our type system is based on Fuzz~\citep{Fuzz}, a family of bounded-linear type
systems for sensitivity analysis originally developed for verifying differential
privacy. To track errors due to rounding, we extend the language with a graded
monadic type $M_u \tau$. Intuitively, $M_u \tau$ is the type of computations
that produce $\tau$ while possibly performing rounding, and $u$ is a real
constant that upper-bounds the rounding error. In this way, we view rounding as
an \emph{effect}, and model rounding computations with a monadic type like
other kinds of computational effects~\citep{DBLP:journals/iandc/Moggi91}. We
interpret our monadic type as a novel graded monad on the category of metric
spaces, which may be of independent interest.

As far as we know, our work is the first type system to provide bounds on
roundoff error. Our type-based approach has several advantages compared to prior
work. First, our system can be instantiated to handle different kinds of error
metrics; our leading application bounds the \emph{relative} error, using a
metric due to~\citet{Olver}. Second, \Lang is an expressive, higher-order
language; by using a primitive operation for rounding, we are able to precisely
describe where rounding is applied. Finally, the analysis in \Lang is
compositional and does not require global optimization.

\paragraph*{Outline of paper.}
After presenting background on floating-point arithmetic and giving an overview
of our system~(\Cref{sec:background}), we present our technical contributions:
\begin{itemize}
  \item The design of \Lang, a language and type system for error
    analysis~(\Cref{sec:language}).
  \item A denotational semantics for \Lang, along with metatheoretic properties
    establishing soundness of the error bound. A key ingredient is the
    \emph{neighborhood monad}, a novel monad on the category of metric spaces
    (\Cref{sec:metatheory}).
  \item A range of case studies showing how to instantiate our language for
    different kinds of error analyses and rounding operations described by the
    floating-point standard. We demonstrate how to use our system to establish
    bounds for various programs through typing (\Cref{sec:examples}).
  \item A prototype implementation for \Lang, capable of inferring types
    capturing roundoff error bounds. We translate a variety of floating-point
    benchmarks into \Lang, and show that our implementation infers error bounds
    that are competitive with error bounds produced by other tools, while often
    running substantially faster (\Cref{sec:implementation}).
  \item Extensions of the neighborhood monad to model more complex rounding
    behavior, e.g., rounding with underflows/overflows, non-deterministic
    rounding, state-dependent rounding, and probabilistic rounding
    (\Cref{sec:extensions}).
\end{itemize}
Finally, we discuss related work (\Cref{sec:relatedwork}) and conclude with
future directions (\Cref{sec:conclusion}).

\section{A Tour of \Lang}\label{sec:background}

\subsection{Floating-Point Arithmetic}
To set the stage, we first recall some basic properties of floating-point
arithmetic. For the interested reader, we point to excellent expositions  by
\citet{Goldberg}, \citet{HighamBook},
and~\citet{boldo_jeannerod_melquiond_muller_2023}.

\ifshort
\else
\begin{table}
\caption{Parameters for floating-point number sets in the IEEE 754-2008
standard. For each, $emin = 1 - emax$.}
\label{tab:formats}
\begin{tabular}{ c c c c }
 \hline
\textbf{Parameter}& \textbf{binary32} & \textbf{binary64} &\textbf{binary128} \\
 \hline
 $p$   & 24    & 53 & 113     \\
$emax$ &   $+127$  & $+1023$ & $+16383$ \\
 \hline
\end{tabular}
\end{table}
\fi

\paragraph{Floating-Point Number Systems}
A floating-point number $x$ in a floating-point number system 
$\F \subseteq \R$  has the form
\begin{equation}
x =  (-1)^s \cdot m \cdot \beta^{e-p+1} \label{eq:fpdef}, 
\end{equation}
where $\beta \in\{ b\in \mathbb{N} \mid b \ge 2\}$ is the \emph{base}, $p \in \{
prec \in \mathbb{N} \mid prec \ge 2\}$ is the \emph{precision}, $m \in
\mathbb{N} \cap [0,\beta^p)$ is the \emph{significand}, $e \in \mathbb{Z} \cap
[\text{emin}, \text{emax}]$ is the \emph{exponent}, and $s \in \{0,1\}$ is the \emph{sign} of
$x$. 
\ifshort
\else
Parameter values for formats defined by the
IEEE floating-point standard~\cite{IEEE2019} are given in \Cref{tab:formats}. 
\fi

Many real numbers cannot be represented exactly in a floating-point format. For example,
the number $0.1$ cannot be exactly represented in {binary64}. Furthermore, the result of most elementary operations on floating-point numbers cannot be represented exactly and must be \emph{rounded} back to a representable value, leading to one of the most distinctive features of floating-point arithmetic: roundoff error.

\ifshort
\else
\begin{table}
\caption{Common Rounding Functions (modes).}
\label{tab:rnd_modes}
\begin{tabular}{ c c c c }
 \hline
\textbf{Rounding mode} & \textbf{Behavior} & \textbf{Notation} & \textbf{Unit Roundoff} \\
\hline
 Round towards $+\infty$   & $\min\{y \in \F \mid y \ge x \}$  & $\rho_{RU}(x)$ & $\beta^{1-p}$   \\
 Round towards $-\infty$   & $\max\{y \in \F \mid y \le x \}$  &   $\rho_{RD}(x)$ & $\beta^{1-p}$ \\
 Round towards $0$   &   $\rho_{RU}(x)$ if $x < 0 $, otherwise  $\rho_{RD}(x)$  & $\rho_{RZ}(x)$
	& $\beta^{1-p} $ \\
 Round towards nearest\tablefootnote{For round towards nearest, there 
	are several possible tie-breaking choices.}   & $\{y \in \F \mid \forall z \in \F,  |x - y| \le |x-z| \} 
	 $   & $\rho_{RN(x})$ & $\frac{1}{2}\beta^{1-p}  $    \\
 \hline
\end{tabular}
\end{table}
\fi

\paragraph{Rounding Operators}
Given a real number $x$ and a floating point format $\F$, a \emph{rounding operator} 
$\rho: \R\rightarrow \F$ is a function that takes $x$ and returns a (nearby) 
floating-point number. The IEEE standard specifies 
that the basic arithmetic operations ($+,-,*,\div,\sqrt{}$) 
behave as if they first computed a correct, infinitely precise result, and then 
rounded the result using one of four rounding 
functions (referred to as modes):  round towards $+\infty$, round towards $- \infty$,
round towards $0$, and round towards nearest (with defined tie-breaking schemes). 
\ifshort
\else
The properties of these operators are given in Figure \ref{tab:rnd_modes}.
\fi

\paragraph{The Standard Model}
By clearly defining floating-point formats and 
rounding functions, the floating-point standard provides a mathematical
model for reasoning about
roundoff error: 
If we write $op$ for an ideal, exact arithmetic operation, and $\widetilde{op}$ for
the correctly rounded, floating-point version of $op$, then
then for any floating-point numbers $x$ and $y$ we have~\cite{HighamBook}
\begin{equation}
  x~\widetilde{op}~y
  \triangleq \rho(x~op~y)
  = (x~op~y)(1+\delta), \quad |\delta| \le u, \quad op\in\{+,-,*,\div\},
  \label{eq:op_model}
\end{equation}
where $\rho$ is an IEEE rounding operator and $u$ is the \emph{unit roundoff}, which is 
upper bounded by $2^{1-p}$ for a binary floating-point format with precision $p$.
\Cref{eq:op_model} is only valid in the absence of underflow and overflow.
We discuss how \Lang can handle underflow and overflow in 
 in \Cref{sec:extensions}.



\paragraph{Absolute and Relative Error}
The most common measures of the accuracy of a floating-point approximation 
$\tilde{x}$ to an exact value $x$ are
\emph{absolute error} ($er_{abs}$) and \emph{relative error} ($er_{rel}$):
\begin{align}
  er_{abs}(x,\tilde{x}) = |\tilde{x}-x| 
	\quad \text{and} \quad
  er_{rel}(x,\tilde{x}) = |(\tilde{x}-x)/x| \quad \text{if}~ x \neq 0. \label{def:ers}
\end{align} 
From \Cref{eq:op_model}, we see that the relative error of the
basic floating-point operations is at most unit roundoff.

The relative and absolute error do not apply uniformly to all values:
the absolute error is well-behaved for small values, while the relative error is well-behaved for large values. 
Alternative error measures that can uniformly represent floating-point error on both large and small values are the \emph{units in the last place} (ULP) error, which measures the number of floating-point values between an approximate and exact value, and its logarithm, \emph{bits of error}~\cite{FPBench}:
\begin{align}
  er_{\textsc{ulp}}(x,\tilde{x}) = |\mathbb{F} \cap [\min(x,\tilde{x}),\max(x,\tilde{x})] | 
	\quad \text{and} \quad
  er_{bits}(x,\tilde{x}) =  \log_2{er_{\textsc{ulp}}(x,\tilde{x})}. \label{def:ulps_ers}
\end{align} 
While static analysis tools that provide sound, worst-case error bound
guarantees for floating-point programs compute relative or absolute error bounds
(or both), the ULP error and its logarithm are often used in tools that optimize
either the performance or accuracy of floating-point programs, like
Herbie~\cite{Herbie} and \textsc{Stoke}~\cite{STOKE-FP}.

\paragraph{Propagation of Rounding Errors}
In addition to bounding the rounding error produced by a floating-point
computation, a comprehensive rounding error analysis must also
quantify how a computation propagates
rounding errors from inputs to outputs. 
The tools
{Rosa}~\cite{Rosa1, Rosa2}, {Fluctuat}~\cite{Fluctuat}, and
{SATIRE}~\cite{SATIRE} account for the propagation 
of rounding errors using Taylor-approximations,
abstract interpretation, and automatic differentiation, respectively. In our work, 
we take a different approach: our language \Lang 
tracks the propagation of
rounding errors using a \emph{sensitivity type system}.

\subsection{Sensitivity Type Systems: An Overview}
The core of our type system is based on Fuzz~\citep{Fuzz}, a family of linear type
systems. The central idea in Fuzz is that each type $\tau$ can be viewed as a
metric space with a \emph{metric} $d_{\tau}$, a notion of distance on values of
type $\tau$. Then, function types describe functions of bounded sensitivity.

\begin{definition}
  A function $f : X \rightarrow Y$ between metric spaces is said to be
  $c$-\emph{sensitive} (or \emph{Lipschitz continuous} with constant $c$)
  \emph{iff} $d_Y(f(x), f(y)) \le c \cdot d_X(x, y)$ for all $x,y \in X$.
  \label{def:csensitive}
\end{definition}

In other words, a function is $c$-sensitive if it can magnify distances between
inputs by a factor of at most $c$. In Fuzz, and in our system, the type $\tau
\lin \sigma$ describes functions that are \emph{non-expansive}, or $1$-sensitive
functions. Intuitively, varying the input of a non-expansive function by
distance $\delta$ cannot change the output by more than distance $\delta$.
Functions that are $r$-sensitive for some constant $r$ are captured by the type
$\bang{r} \tau \lin \sigma$; the type $\bang{r} \tau$ scales the metric of
$\tau$ by a factor of $r$.

To get a sense of how the type system works, we first introduce a metric on real
numbers proposed by \citet{Olver} to capture relative error in numerical
analysis.

\begin{definition}[The Relative Precision (RP) Metric]
Let $\tilde{x}$ and $x$ be nonzero real numbers of the same
sign. Then the {relative precision} (RP) of $\tilde{x}$ as an approximation to $x$ is given by
\begin{equation} RP(x,\tilde{x}) = |ln(x/\tilde{x})| . \end{equation}
\label{def:rp}
\end{definition}

While \Cref{def:rp} is a true metric, satisfying the usual axioms of zero-self distance, symmetry, and the triangle inequality, the relative error (\Cref{def:ers}) and the ULP error (\Cref{def:ulps_ers}) are not.
Rewriting \Cref{def:rp} and the relative error as follows, and by considering the Taylor expansion of the exponential function, 
we can see that the relative precision is a close approximation to the relative error so long as
$\delta \ll 1$:
\begin{align}
er_{rel}(x,\tilde{x}) &= |\delta|; \quad \tilde{x} = (1+\delta)x,  \label{eq:RPbnd1}\\
RP(x,\tilde{x}) &= |\delta|; \quad \tilde{x} = e^\delta x. \label{eq:RPbnd2}
\end{align} 
Moreover, if $\tilde{x}$ as an approximation to $x$ of relative precision $0 \le \alpha < 1$, then $\tilde{x}$ approximates $x$ with relative error~\citep[cf.][eq 3.28 95]{Olver}
\begin{equation}
  \epsilon = e^\alpha -1 \le \alpha/(1-\alpha). \label{eq:conv}
\end{equation}

In \Lang, we can write a function $\mathbf{pow2}$ that squares its argument and
assign it the following type:
\begin{align*}
\mathbf{pow2} &\triangleq \lambda x.  ~\mathbf{mul}~(x,x) : ~!_2 \num \multimap \num. 
\end{align*}
The type $\num$ is the numeric type in \Lang; for now, we can think of it as
just the ideal real numbers $\R$. The type $\bang{2} \num \lin \num$ states
that $\mathbf{pow2}$ is $2$-sensitive under the RP metric. More generally, this type reflects that the sensitivity of a function to its inputs is dependent on how many times the input is used in the function body. Now, spelling all of  this out, if we have two inputs $v$ and $v \cdot e^\delta$ at distance $\delta$ under the RP metric, then applying $\mathbf{pow2}$ leads to outputs $v^2$ and $(v \cdot e^\delta)^2 = v^2 \cdot e^{2 \cdot \delta}$, which are at distance (at most) $2 \cdot \delta$ under the RP metric.

\subsection{Roundoff Error Analysis in \Lang: A Motivating Example}
So far, we have not considered roundoff error: $\mathbf{pow2}$ simply squares
its argument without performing any rounding.  Next, we give an idea of how
rounding is modeled in \Lang, and how sensitivity interacts with roundoff error.

The purpose of a rounding error analysis is to derive an a priori bound on the
effects of rounding errors on an algorithm~\cite{HighamBook}.  Suppose we are
tasked with performing a rounding error analysis on a function $\mathbf{pow2'}$,
which squares a real number and rounds the result.  Using the standard model for
floating-point arithmetic (\Cref{eq:op_model}), the analysis is simple: the
result of the function is 
\begin{align}
  \mathbf{pow2'}(x) &= \rho(x*x) = (x*x)(1+\rnderr), \quad |\rnderr| \le u \label{eq:ex_pow2}, 
\end{align}
and the relative error is bounded by the unit roundoff, $u$. Our insight is 
that a type system can be used to perform this analysis, by
modeling rounding as an \emph{error producing} effectful operation.

To see how this works, the function $\mathbf{pow2'}$ can be defined in \Lang as
follows:
\begin{align*}
\mathbf{pow2'} &\triangleq \lambda x.  ~\rnd (\mathbf{mul}~(x,x)) : ~!_2 \num \multimap M_u \num. 
\end{align*}
Here, $\mathbf{rnd}$ is an effectful operation that applies rounding to its
argument and produces values of monadic type $M_\rnderr \num$; intuitively, this
type describes computations that produce numeric results while performing
rounding, and incurring at most $\rnderr$ in rounding error. Thus, our type
for $\mathbf{pow2'}$ captures the desired error bound from \Cref{eq:ex_pow2}:
for any input $v \in \R$, $\mathbf{pow2'}(v)$ approximates its ideal, infinitely
precise counterpart $\mathbf{pow2}(v)$ within RP distance at most $u$, the unit
roundoff.

To formalize this guarantee, programs of type $M_\rnderr \tau$ can be executed
in two ways: under an ideal semantics where rounding operations act as the
identity function, and under an approximate (floating-point) semantics where
rounding operations round their arguments following some prescribed rounding
strategy. We formalize these semantics in \Cref{sec:metatheory}, and show our
main soundness result: for programs with monadic type $M_\rnderr\num$, the
result of the ideal computation differs from the result of the approximate
computation by at most $\rnderr$.

\paragraph{Composing Error Bounds}
The type for $\mathbf{pow2'} : \bang{2} \num \lin M_u \num$ actually guarantees
a bit more than just a bound on the roundoff: it also guarantees that the
function is $2$-sensitive under the \emph{ideal} semantics, just like for
$\mathbf{pow2}$. (Under the approximate semantics, on the other hand, the
function does \emph{not} necessarily enjoy this guarantee.) It turns out that
this added piece of information is crucial when analyzing how rounding errors
propagate.

To see why, suppose we define a function that maps any number $v$ to its fourth
power: $v^4$. We can implement this function by using $\mathbf{pow2'}$ twice,
like so:
\begin{align*}
\mathbf{pow4} ~x &\triangleq \letbind ~y = \mathbf{pow2'} ~x ~\tin 
	\mathbf{pow2'}~y:  M_{3u} \num. 
\end{align*}
The $\letbind - \tin - $ construct sequentially composes two monadic, effectful
computations; to keep this example readable, we have elided some of the other
syntax in \Lang. Thus, $\mathbf{pow4}$ first squares its argument, rounds the
result, then squares again, rounding a second time.

The bound $3u$ on the total roundoff error deserves some explanation. In the
typing rules for \Lang we will see in \Cref{sec:language}, this index is
computed as the sum $2u + u$, where the first term $2u$ is the error $u$ from the
first rounding operation \emph{amplified by $2$ since this error is fed into the
second call of $\mathbf{pow2'}$, a $2$-sensitive function}, and the second
term $u$ is the roundoff error from the second rounding operation.  If we think
of $\mathbf{pow2'}$ as mapping a numeric value $a$ to a pair of outputs $(b,
\tilde{b})$, where $b$ is the result under the exact semantics and $\tilde{b}$
is the result under the approximate semantics, we can visualize the computation
$\mathbf{pow4}(a)$ as the following composition:
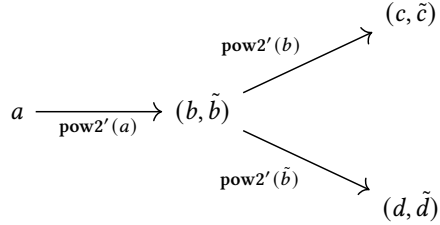
\[\begin{tikzcd}
	&&&& {(c, \tilde{c})} \\
	a && {(b, \tilde{b})} \\
	&&&& {(d, \tilde{d})}
	\arrow["{{\mathbf{pow2'}}(a)}"', from=2-1, to=2-3]
	\arrow["{{\mathbf{pow2'}}(b)}", from=2-3, to=1-5]
	\arrow["{{\mathbf{pow2'}}(\tilde{b})}"', from=2-3, to=3-5]
\end{tikzcd}\]
From left-to-right, the ideal and approximate results of $\mathbf{pow2'}(a)$ are
$b$ and $\tilde{b}$, respectively; the grade $u$ on the monadic return type of
$\mathbf{pow2'}$ ensures that these values are at distance at most $u$. The
ideal result of $\mathbf{pow4}(a)$ is $c$, while the approximate result of
$\mathbf{pow4}(a)$ is $\tilde{d}$. (The values $\tilde{c}$ and $d$ arise from
mixing ideal and approximate computations, and do not fully correspond to either
the ideal or approximate semantics.) The $2$-sensitivity guarantee of
$\mathbf{pow2'}$ ensures that the distance between $c$ and $d$ is at most twice
the distance between $b$ and $\tilde{b}$---leading to the $2u$ term in the
error---while the distance between $d$ and $\tilde{d}$ is at most $u$. By
applying the triangle inequality, the overall error bound is at most $2u + u =
3u$.

\paragraph{Error Propagation} 
So far, we have described how to bound the rounding error of a single
computation applied to a single input. In practice, it is also often useful to
analyze how errors \emph{propagate} through a computation: given inputs with
some roundoff error $u$, how does the output roundoff error depend on $u$? 
Our system can also be used for this kind of analysis; we give detailed examples in Section 5, but can use our running example of $\mathbf{pow4}$ to illustrate the idea here

If we denote by $f$ the interpretation of an infinitely precise program, and by
$g$ an interpretation of a finite-precision program that approximates $f$, then
we can use the triangle inequality to derive an upper bound on the relative
precision of $g$ with respect to $f$ on distinct inputs: 
\begin{equation} 
RP(g({x}), f(y)) \le RP(g({x}), f({x})) + RP({f}({x}), f(y)). \label{eq:erprop}
\end{equation} 
\noindent This upper bound is the sum of two terms: the first 
reflects the \emph{local} 
rounding error---how much error is produced by the approximate function, 
and the second reflects by how much the function magnifies 
errors in the inputs---the sensitivity of the function. 

Now, given that the full signature of $\mathbf{pow4}$ is 
$\mathbf{pow4}:~ !_4 \num \multimap M_{3u}\num$,
if we denote by $g$ and $f$ the approximate and 
ideal interpretations of $\mathbf{pow4}$, from 
\Cref{eq:erprop} we expect our type system to 
produce the following bound on the propagation of 
error in $\mathbf{pow4}$. For any exact number $x$ and its approximation
$\tilde{x}$ at distance at most $u'$, we have:
\begin{equation} 
  RP(g({\tilde{x}}), f(x)) \le 3u + 4u' .
\end{equation} 
\noindent The term $4u'$ reflects that $\mathbf{pow4}$ is $4$-sensitive in its
argument, and that the (approximate) input value $\tilde{x}$ differs from its
ideal value $x$ by at most $u'$. In fact, we can use $\mathbf{pow4}$ to
implement a function $\mathbf{pow4'}$ with the following type:
\[
  \mathbf{pow4'} : M_{u'} \num \lin M_{3u + 4u'} \num .
\]
The type describes the error propagation: roundoff error at most $u'$ in the
input leads to roundoff error at most $3u + 4u'$ in the output.


\section{The Language \Lang}\label{sec:language}  

\subsection{Syntax} 

\begin{figure}[tbp]
  \begin{alignat*}{3}
         &\text{Types } \sigma, \tau &::=~ &\mathbf{unit}
         \mid \num
         \mid \sigma \times \tau 
         \mid \sigma \otimes \tau
         \mid \sigma + \tau 
         \mid \sigma \multimap \tau
         \mid {\bang{s} \sigma}
         \mid {M_u \tau}
         \\
         &\text{Values } v, w \ &::=~ &x
         \mid \langle \rangle
         \mid k \in R
         \mid \langle v,w \rangle 
         \mid  (v, w)
         \mid \inl \ v
         \mid \inr \ v
         \\
         & & & \mid \lambda x.~e 
         \mid [v]
         \mid\rnd v
         \mid {\ret v} 
         \mid \letbind(\rnd v,x.f) \\
         &\text{Terms } e, f &::=~ & v
	\mid v~w
         \mid {\pi}_i\ v
         \mid \mathbf{let} \ (x,y) = v \ \tin e
         \mid \mathbf{case} \ v \ \mathbf{of} \ (\inl x.e \ | \ \inr x.f) 
         \\
         & & & \mid
          \tlet [x]  = v \ \tin  e
         \mid {\letbind (v,x.f)} 
         \mid \tlet x  = e \ \tin  f 
         \mid \mathbf{op}(v) \quad \mathbf{op} \in \mathcal{O}
  \end{alignat*}
  \caption{Types, values, and terms.}
  \label{fig:syntax}
\end{figure}

\Cref{fig:syntax} presents the syntax of types and terms.  \Lang~ is based on
Fuzz \citep{Fuzz}, a linear call-by-value $\lambda$-calculus.  For simplicity we
do not treat recursive types, and \Lang~does not have general recursion.

\paragraph{Types}
Some of the types in \Cref{fig:syntax} have already been mentioned in 
\Cref{sec:background}, including the 
linear function type $\tau \multimap \sigma$, 
the metric scaled $!_s\sigma$ type, and the monadic $M_\rnderr\num$
type. The base types are $\mathbf{unit}$ and numbers $\num$.
Following Fuzz, \Lang has sum types $\sigma + \tau$
and two product types, $\tau \tensor
\sigma$ and $\tau \times \sigma$, which are interpreted as pairs with
different metrics.

\paragraph{Values and Terms}

Our language requires that all computations are explicitly sequenced by
let-bindings, ${\tlet x~= v ~\tin e}$, and term constructors and eliminators are
restricted to values (including variables). This refinement of Fuzz
better supports extensions to effectful languages~\citep{DLG}.
In order to sequence monadic and metric 
scaled types, \Lang provides the eliminators $\letbind(v,x.e)$
and $\tlet [x] = v \ \tin e$, respectively. The constructs 
 $\rnd v$ and $\ret v$ lift values of plain type to monadic type; for 
metric types, the construct $[ v ]$ indicates scaling the 
metric of the type by a constant.

\Lang is parameterized by a set $R$ of numeric constants with type
$\num$. In \Cref{sec:examples}, we will instantiate $R$ and interpret $\num$
as a concrete set of numbers with a particular metric.  \Lang is also
parameterized by a signature $\Sigma$: a set of operation symbols $\mathbf{op}
\in \mathcal{O}$, each with a type $\sigma \lin \tau$, and a
function $op : CV(\sigma) \rightarrow CV(\tau)$ mapping closed values of type
$\sigma$ to closed values of type $\tau$. We write $\{ \mathbf{op} : \sigma \multimap
\tau\}$ in place of the tuple $( \sigma \multimap \tau , op : CV(\sigma) \rightarrow
CV(\tau), \mathbf{op})$. For now, we make no assumptions on the functions $op$;
in \Cref{sec:metatheory} we will need further assumptions.  

\subsection{Static Semantics}\label{sec:staticsemantics} \begin{figure}
\begin{center}
\AXC{$s \ge 1$}
\RightLabel{(Var)}
\UIC{$\Gamma, x:_s \sigma, \Delta \vdash x : \sigma$}
\bottomAlignProof
\DisplayProof
\hskip 0.5em
\AXC{$\Gamma, x:_1 \sigma \vdash e : \tau$}
\RightLabel{($\multimap$ I)}
\UnaryInfC{$\Gamma \vdash \lambda x. e : \sigma \multimap \tau $}
\bottomAlignProof
\DisplayProof
\hskip 0.5em
\AXC{$\Gamma \vdash v : \sigma \multimap \tau$}
\AXC{$\Theta \vdash w : \sigma $}
\RightLabel{($\multimap$ E)}
\BinaryInfC{$\Gamma + \Theta \vdash vw : \tau $}
\bottomAlignProof
\DisplayProof
\vskip 1em

\AXC{}
\RightLabel{(Unit)}
\UIC{$\Gamma \vdash \langle \rangle : \mathbf{unit}$}
\bottomAlignProof
\DisplayProof
\hskip 0.5em
\AXC{$\Gamma \vdash v : \sigma$}
\AXC{$\Gamma \vdash w : \tau$}
\RightLabel{($\times$ I)}
\BinaryInfC{$\Gamma \vdash \langle v, w \rangle: \sigma \times \tau $}
\bottomAlignProof
\DisplayProof
\hskip 0.5em
\AXC{$\Gamma \vdash v : \tau_1 \times \tau_2$}
\RightLabel{($\times$ E)}
\UIC{$\Gamma \vdash {\pi}_i \ v : \tau_i$}
\bottomAlignProof
\DisplayProof
\vskip 1em

\AXC{$\Gamma \vdash v : \sigma $}
\AXC{$\Theta \vdash w : \tau$}
\RightLabel{($\tensor$ I)}
\BIC{$\Gamma + \Theta \vdash (v, w) : \sigma \tensor \tau$}
\bottomAlignProof
\DisplayProof
\hskip 0.5em
\AXC{$\Gamma \vdash v : \sigma \tensor \tau$ }
\AXC{$\Theta,x:_s \sigma,y:_s\tau \vdash e: \rho $}
\RightLabel{($\tensor$ E)}
\BIC{$s * \Gamma + \Theta \vdash \tlet (x,y) \ = \ v \ \tin e : \rho $}
\bottomAlignProof
\DisplayProof
\vskip 1em


\AXC{$\Gamma \vdash v : \sigma$ }
\RightLabel{($+$ $\text{I}_L$)}
\UIC{$\Gamma \vdash \inl \ v : \sigma + \tau$}
\bottomAlignProof
\DisplayProof
\hskip 0.5em
\AXC{$\Gamma \vdash v : \tau$ }
\RightLabel{($+$ $\text{I}_R$)}
\UIC{$\Gamma \vdash \inr \ v : \sigma + \tau$}
\bottomAlignProof
\DisplayProof
\hskip 0.5em
\AXC{$\Gamma \vdash v : {!_s \sigma}$}
\AXC{$\Theta, x:_{t*s} \sigma \vdash e : \tau$}
\RightLabel{($!$ E)}
\BIC{$t * \Gamma + \Theta \vdash \tlet [x] = v \ \tin e : \tau$}
\bottomAlignProof
\DisplayProof
\vskip 1em


\AXC{$\Gamma \vdash v : \sigma+\tau$}
\AXC{$\Theta, x:_s \sigma \vdash e : \rho$ \qquad
$\Theta, x:_s \tau \vdash f: \rho$}
\AXC{$s > 0$}
\RightLabel{($+$ E)}
\TIC{$s * \Gamma + \Theta \vdash \mathbf{case} \ v \ \mathbf{of} \ (\inl x.e \ | \ \inr x.f) : \rho$}
\bottomAlignProof
\DisplayProof
\hskip 0.5em
\AXC{$\Gamma \vdash v : \sigma$ }
\RightLabel{($!$ I)}
\UIC{$s * \Gamma \vdash [v] : {!_s \sigma}$}
\bottomAlignProof
\DisplayProof
\vskip 1em


\AXC{$\Gamma \vdash e :  \tau$}
\AXC{$\Theta, x:_{s} \tau \vdash f : \sigma$}
\AXC{$s > 0$}
\RightLabel{(Let)}
\TIC{$s * \Gamma + \Theta \vdash \tlet x = e \ \tin f : \sigma$}
\bottomAlignProof
\DisplayProof
\hskip 0.5em
\AXC{$k \in R$}
\RightLabel{(Const)}
\UIC{$\Gamma \vdash k : \num$}
\bottomAlignProof
\DisplayProof
\hskip 0.5em
\vskip 1em


\AXC{$\Gamma \vdash e :  M_q \tau$}
\AXC{$r \ge q$}
\RightLabel{(Subsumption)}
\BIC{$\Gamma \vdash e :  M_{r} \tau$}
\bottomAlignProof
\DisplayProof
\hskip 0.5em
\AXC{$\Gamma \vdash v : \tau$}
\RightLabel{(Ret)}
\UIC{$\Gamma \vdash \ret v : M_0 \tau$}
\bottomAlignProof
\DisplayProof
\hskip 0.5em
\AXC{$\Gamma \vdash v : \num$}
\RightLabel{(Rnd)}
\UIC{$\Gamma \vdash \rnd \ v : M_q \num$}
\bottomAlignProof
\DisplayProof
\vskip 1em


\AXC{$\Gamma \vdash v : M_r \sigma$}
\AXC{$\Theta, x:_{s} \sigma \vdash f : M_{q} \tau$}
\RightLabel{($M_u$ E)}
\BIC{$s * \Gamma + \Theta \vdash \letbind(v,x.f) : M_{s*r+q} \tau$}
\bottomAlignProof
\DisplayProof
\hskip 0.5em
\AXC{$\Gamma \vdash v : \sigma$}
\AXC{$\{ \mathbf{op} :\sigma \lin \num \} \in {\Sigma}$}
\RightLabel{(Op)}
\BIC{$\Gamma \vdash \mathbf{op}(v) : \num$}
\bottomAlignProof
\DisplayProof

\end{center}
    \caption{Typing rules for \Lang, with $s,t,q,r \in \NNR \cup \{\infty\}$. }
    \label{fig:typing_rules}
\end{figure}

The static semantics of \Lang~ is given in Figure \ref{fig:typing_rules}. 
Before stepping through the details of each rule, we require some 
definitions regarding typing judgments and typing environments. 
 
Terms in \Lang are typed with judgments of the form 
$\Gamma \vdash e : \sigma$ where $\Gamma$
is a typing environment and $\sigma$ is a type.
Environments are defined by
the syntax $\Gamma, \Delta ::= \cdot \mid \Gamma, x:_s\sigma$.
We can also view a typing environment $\Gamma$ as a partial map from variables
to types and sensitivities, where $(\sigma, s) = \Gamma(x)$ when ${x:_s\sigma \in
\Gamma}$. Intuitively, if the environment $\Gamma$ has a binding ${x:_s\sigma \in \Gamma}$, then
a term $e$ typed under $\Gamma$ has sensitivity $s$ to perturbations in the variable
$x$; $0$-sensitivity means that the term does not depend on $x$, while infinite
sensitivity means that any perturbation in $x$ can lead to
arbitrarily large changes in $e$. Well-typed expressions of the form  $x :_s \sigma \vdash e : \tau$ represent computations that have permission to be $s$-sensitive in the variable $x$.

Many of the typing rules for \Lang~ involve summing  and scaling typing
environments. The notation $s * \Gamma$ denotes scalar multiplication of the
variable sensitivities in $\Gamma$ by $s$, and is defined as 
\begin{equation*}
\begin{aligned}[c]
s * \cdot &= \cdot
\end{aligned}
\qquad\qquad
\begin{aligned}[c]
s * (\Gamma,x:_t \sigma) &= s * \Gamma,x:_{s*t},
\end{aligned}
\end{equation*}
where we require that $0 \cdot \infty = \infty \cdot 0 = 0$.  The sum $\Gamma +
\Delta$ of two typing environments is defined if they assign the same types to
variables that appear in both environments. All typing rules that involve summing environments ($\Gamma + \Delta$) implicitly
require that $\Gamma$ and $\Delta$ are \emph{summable}.
\begin{definition}
 The environments $\Gamma $ and $\Delta$ are \emph{summable}  \emph{iff} for any ${x \in dom(\Gamma) \cap dom(\Delta)}$, if
   $(\sigma,s) = \Gamma(x)$, then there exists an element $t \in \NNR \cup \{ \infty\}$ such that $(\sigma,t) =
  \Delta(x)$. 
\end{definition}
\noindent Under this condition, we can define the sum $\Gamma + \Delta$ as follows. 
\begin{equation*}
\begin{aligned}[c]
\cdot + \cdot &= \cdot \\
\Gamma + (\Delta, x:_s
\sigma) &= (\Gamma + \Delta), x:_s\sigma \text{ if } x \notin \Gamma
\end{aligned}
\qquad\qquad
\begin{aligned}[c]
(\Gamma,x:_s \sigma) + \Delta &= (\Gamma + \Delta),x:_s\sigma \text{ if } x
\notin \Delta \\
(\Gamma,x:_s \sigma) + (\Delta,x:_t \sigma) &= 
  (\Gamma + \Delta), x:_{s+t}\sigma \\
\end{aligned}
\end{equation*}

We now consider
the rules in \Cref{fig:typing_rules}. The simplest are ({Const}) and
({Var}), which allow any constant to be used under any environment, and
allow a variable from the environment to be used so long as its sensitivity is at
least $1$.

The introduction and elimination rules for the products $\tensor$ and $\times$
are similar to those given in Fuzz. In ($\tensor$ I), introducing the pair
requires summing the environments in which the individual elements were defined,
while in ($\times~$ I), the elements of the pair share the same environment. 

The typing rules for sequencing (Let) and case analysis
($+$ E) both require that the sensitivity $s$ is strictly
positive. While the restriction in ({Let}) is not needed for a
terminating calculus, like ours, it is required for soundness in the presence of
non-termination~\citep{DLG}. The restriction in ($+$ E)  is needed for
soundness (we discuss this detail in \Cref{sec:relatedwork}).

The remaining interesting rules are those for metric scaling and monadic 
types. In the ($!$ I) rule, the box constructor 
$[ {-} ]$ indicates scalar multiplication of an environment. 
The ($!$ E) rule is
similar to ($\tensor$ {E}), but includes the
scaling on the let-bound variable. 

The rules (Subsumption), ({Ret}), (Rnd), and
($M_u$ E) are the core rules for performing rounding error analysis in
\Lang. Intuitively, the monadic type $M_\rnderr \num$ describes computations
that produce numeric results  while performing rounding, and incur at most
$\rnderr$ in rounding error. The subsumption rule
 states that rounding error bounds can be loosened.
The ({Ret}) rule states that we can lift terms of plain type to
monadic type without introducing rounding error. The (Rnd) rule
types the primitive rounding operation, which introduces roundoff errors.
Here, $q$ is a fixed numeric constant describing the roundoff error incurred by
a rounding operation. The precise value of this constant depends on the precision
of the format and the specified
rounding mode; we leave $q$ unspecified for now. In \Cref{sec:examples}, we
will illustrate how to instantiate our language to different settings.

The monadic elimination rule ($M_u$ E) allows sequencing two rounded computations
together. This rule formalizes the interaction between sensitivities and
rounding, as we illustrated in \Cref{sec:background}: the rounding error of the
body of the let-binding $\letbind(v,x.f)$ is upper bounded by the sum of the
roundoff error of the value $v$ scaled by the sensitivity of $f$ to $x$, and the
roundoff error of $f$.

\ifshort
\else
Before introducing our dynamic semantics, we note that
the static semantics of \Lang~
enjoy the properties of weakening and substitution, and 
define the notion
of a subenvironment.

\begin{definition}[Subenvironment]
$\Delta$ is a \emph{subenvironment} of $\Gamma$, written 
$\Delta \sqsubseteq \Gamma$, if whenever $\Delta(x) = (s,\sigma)$
for some sensitivity $s$ and type $\sigma$,  then there
exists a sensitivity $s'$ such that $s' \ge s$ and
$\Gamma(x) = (s',\sigma)$. \label{def:subenv}
\end{definition}

\begin{lemma}[Weakening] Let $\Gamma \vdash e : \tau$ be a well-typed term. Then
for any typing environment $\Delta \sqsubseteq \Gamma$, there is a derivation of $\Delta
\vdash e : \tau$. \label{thm:weakening} \end{lemma}

\begin{proof} By induction on the typing derivation of $\Gamma \vdash e: \tau$.
\end{proof}

\begin{lemma}[Substitution] Let $\Gamma, \Delta \vdash e : \tau$ be a well-typed
term, and let $\vec{v} : \Delta$ be a well-typed substitution of closed values,
i.e., we have derivations $\vdash v_x : \Delta(x)$. Then there is a derivation
of \[ \Gamma \vdash e[\vec{v}/dom(\Delta)] : \tau. \] \label{thm:substitution}
\end{lemma}
\begin{proof} The base cases (Unit), (Const), and (Var)
follow easily, and the remaining of the cases follow by applying the induction
hypothesis to every premise of the relevant typing rule. \end{proof}
\fi

\subsection{Dynamic Semantics} 
We use a small-step operational semantics adapted
from Fuzz~\citep{Fuzz}, extended with rules for the monadic let-binding. 
\ifshort
\else
The complete set of evaluation rules is given in Figure
\ref{fig:eval_rules}, where the judgment $e \mapsto e'$ indicates that the expression
$e$ takes a single step, resulting in the expression $e'$. 
\begin{figure}
\begin{center}

\begin{equation*}
\begin{aligned}[c]
	\pi_i\langle v_1,v_2 \rangle &\mapsto v_i \\
	\mathbf{op}(v) &\mapsto op(v)\\
	(\lambda x.e) \ v &\mapsto e[v/x] 
\end{aligned}
\quad
\begin{aligned}[c]
  \tlet (x, y) = (v, w)\ \tin e &\mapsto e[v/x][w/y] \\
	\tlet [x] = [v] \ \tin e &\mapsto e[v/x] \\
	\letbind(\ret v, x.e) &\mapsto e[v/x] 
\end{aligned}
\end{equation*}
\vskip -1em
\begin{align*}
	\qquad \mathbf{case} \ (\mathbf{in}k \ v) \ \mathbf{of} \ (\mathbf{inl} \ x.e_l \ | \ \mathbf{inr} \ x.e_r )  \mapsto e_k[v/x]
  \qquad\qquad(k \in \{l, r \})
\end{align*}
\vskip -1.75em
\begin{align*}
	\letbind(\letbind(v,x.f),y.g) &\mapsto \letbind(v,x.\letbind(f,y.g)) \quad x\notin FV(g)
\end{align*}
\vskip -0.25em

	\AXC{$e \mapsto e'$}
	\UIC{$\tlet x = e \ \tin f \mapsto \tlet x = e' \ \tin f$}
	\DisplayProof

\end{center}
    \caption{Evaluation rules for \Lang.}
    \label{fig:eval_rules}
\end{figure}

\fi

Although our language does not have recursive types, the $\letbind$ construct
makes it somewhat less obvious that the calculus is terminating: the evaluation
rules for $\letbind$ rearrange the term but do not reduce its size.  %
\ifshort
\else
Even so,
it is possible to show that well-typed programs are terminating.
If we denote the set of closed values of type $\tau$ by $CV(\tau)$ and the set
of closed terms of type $\tau$ by $CT(\tau)$ so that $CV_\tau \subseteq
CT(\tau)$, and define $\mapsto^*$ as the reflexive, transitive closure of the
single step judgment $\mapsto$, then we can state our termination theorem as
follows.

\begin{theorem} [Termination] If $\cdot \vdash e : \tau$ then there exists
$v \in CV(\tau)$ such that $e \mapsto^* v$. \label{thm:SN} \end{theorem}

The proof of \Cref{thm:SN} follows by a standard logical relations argument.
Our logical relations will also show something a bit stronger: roughly speaking,
programs of monadic type $M_u \tau$ can be viewed as finite-depth trees.
\begin{definition} We define the reducibility predicate $\mathcal{R}_\tau$
inductively on types in Figure \ref{fig:redpred}. \label{def:redpred}
\end{definition} \begin{figure} \begin{alignat*}{4} \mathcal{R}_\tau &\triangleq
\{e \mid e \in CT(\tau) \tand \exists v \in CV(\tau). \ e \mapsto^* v \tand v
\in \mathcal{VR}_\tau \} & & \\ \mathcal{VR}_{unit} &\triangleq \{ \langle
\rangle \} & \mathcal{VR}_{\num} &\triangleq R \\ \mathcal{VR}_{\sigma \times
\tau} &\triangleq \{ \langle v,w \rangle \mid v \in \mathcal{R}_\sigma \tand w
\in \mathcal{R}_\tau \} & \mathcal{VR}_{\sigma \tensor \tau} &\triangleq \{ (v, w)
\mid v \in \mathcal{R}_\sigma \tand w \in \mathcal{R}_\tau \} \\
\mathcal{VR}_{\sigma + \tau} &\triangleq \{ v \mid \inl v \ \& \ v \in
\mathcal{R}_\sigma \text{ or } \inr v \tand v \in \mathcal{R}_\tau \} &
\mathcal{VR}_{\sigma \multimap \tau} & \triangleq \{ v \mid \forall w \in
\mathcal{VR}_\sigma.\ vw \in \mathcal{R}_\tau \} \\ \mathcal{VR}_{\bang{s} \tau} &
\triangleq \{ [v] \mid v \in \mathcal{R}_\tau \} & \mathcal{VR}_{M_u \tau} &
\triangleq \bigcup_{n} \mathcal{VR}^n_{M_u \tau} \\ \mathcal{VR}^0_{M_u \tau} &
\triangleq \{ v \mid v \equiv \ret w \ \& \ w \in \mathcal{R}_\tau \text{ or } v
\equiv \rnd \ k \tand k \in \mathcal{R}_{\num}\} & & \end{alignat*}
\begin{flalign*} &\mathcal{VR}^{n+1}_{M_u \tau} \triangleq \mathcal{R}^n_{M_u
\tau} \cup \Bigg\{ \letbind(v,x.f) \mid \exists \ \sigma, u_1, u_2, j. \ u \ge
u_1 + u_2 \tand n > j \tand v \in \mathcal{VR}^j_{M_{u_1} \sigma} \Bigg. &
\\ & \Bigg. \hskip 24 em \tand \left(\forall w\in \mathcal{VR}_\sigma, f[w/x] \in
\mathcal{VR}^{n-j}_{M_{u_2} \tau}\right) \Bigg\} & \end{flalign*}
\caption{Reducibility Predicate} \label{fig:redpred} \end{figure}

The proof of  \Cref{thm:SN} requires two lemmas. The first is quite
standard, and the definition of the reducibility predicate ensures that the
proof follows without complication by induction on the derivation $\cdot \vdash
e : \tau$.

\begin{theorem} The predicate $\mathcal{VR}$ is preserved by backward and
forward reductions: if $\cdot \vdash e : \tau$ and $e \mapsto e'$ then $e \in
\mathcal{R}_\tau \iff e' \in \mathcal{R}_\tau$. \label{thm:predpres}
\end{theorem}

Observe that the following lemma follows without complication by induction on
$m$, the depth of the predicate $\mathcal{VR}$.

\begin{lemma}[Subsumption] If $ e \in \mathcal{VR}^m_{M_u \tau}$ then $e \in
\mathcal{VR}^m_{M_{u'} \tau}$ supposing $u' \ge u$. \label{thm:subsumption}
\end{lemma}

Using \Cref{thm:predpres} and \Cref{thm:subsumption} we can prove the
following.

\begin{lemma} If $\cdot \vdash e : \tau$ then $e \in \mathcal{R}_\tau$.
\label{thm:SN1}
\end{lemma}

\begin{proof} We first prove the following stronger statement. Let $\Gamma
\triangleq x_1 :_{s_1} \sigma_1, \cdots, x_i :_{s_i} \sigma_i$ be a typing
environment and let $\vec{w}$ denote the values $\vec{w} \triangleq w_1, \cdots,
w_i$. If {$\Gamma \vdash e : \tau$} and {$w_i \in \mathcal{VR}_{\Gamma(x_i)}$}
for every {$x_i \in dom(\Gamma)$} then {$e[\vec{w}/dom(\Gamma)] \in
\mathcal{R}_\tau$}. The proof follows by induction on the derivation $\Gamma
\vdash e : \tau$. Let us consider the monadic cases, as the non-monadic cases
are standard. The base cases (Const), (Ret), and (Rnd)
follow by definition, and \textsc{Subsum} follows by theorem
\ref{thm:subsumption}. The case for ($M_u$ E) requires some detail. The
rule is

\begin{center} \vskip 1em \noLine \AXC{($M_u$ E)} \kernHyps{-2em} \UIC{}
\noLine \UIC{$\Gamma \vdash v : M_u \sigma$} \AXC{$\Theta, x:_{s} \sigma \vdash
f : M_{u'} \tau$} \BIC{$s * \Gamma + \Theta \vdash \letbind(v,x.f) : M_{s*u+u'}
\tau$} \DisplayProof \vskip 1em \end{center} and so we are required to show
{$\letbind(v, x.f) \in \mathcal{R}_{M_{s*u + u'}\tau}$}. Let us denote the
typing environment $s * \Gamma + \Theta$ by $\Delta$. We apply weakening (theorem
\ref{thm:weakening}) and the induction hypothesis to each premise of the rule to
obtain ${f[\vec{w}/dom(\Delta)][w'/x] \in \mathcal{R}_{M_{u'}\tau}}$ for any $w'
\in R_\sigma$ and {$v[\vec{w}/dom(\Delta)] \in \mathcal{R}_{u}\sigma$}. The
proof follows by cases on $v$. We will make use of the following lemma.

For any typing environment $\Gamma$, let $\Gamma e : \tau$ be a well-typed
term,and let $\vec{v}_1, \vec{v}_2 : \Gamma$ be a well-typed substitutions of
closed values. If $e[\vec{v}_1/dom(\Gamma)]\in \mathcal{VR}^m_{M_u\tau}$ for
some $m\in\mathbb{N}$ and $e[\vec{v}_2/dom(\Gamma)]\in \mathcal{R}_{M_u\tau}$,
then $e[\vec{v}_2/dom(\Gamma)]\in \mathcal{VR}^m_{M_u\tau}$.

Case: $v \equiv \rnd \ k$ with $k \in \R$. From lemmas \ref{thm:predpres} and
\ref{thm:substitution} we have that ${f[\vec{w}/dom(\Delta)][r/x] \in
\mathcal{VR}^m_{M_{u'}\tau}}$ for some $r \in \R$ and $m \in \mathbb{N}$, and
it follows that {$\letbind(v, x.f) \in \mathcal{VR}^{m+1}_{M_{s*u + u'}\tau}$}.
 
Case: $v \equiv \ret v'$ with $v' \in \mathcal{VR}_{\sigma}$. From the
reduction rules we have that \[{\letbind(v, x.f)[\vec{w}/dom(\Delta)] \mapsto
f[v'/x][\vec{w}/dom(\Delta)]}.\] Instantiating the induction hypothesis with
$v'$ yields ${f[v'/x][\vec{w}/dom(\Delta)] \in \mathcal{R}_{M_{u'}\tau}}$ and
the conclusion follows from lemmas \ref{thm:subsumption}, \ref{thm:predpres},
and \ref{thm:substitution}.

Case: $v \equiv \letbind(\rnd~ k, y.g)$. From the reduction rules we have
that\[\letbind(v, x.f)[\vec{w}/dom(\Delta)] \mapsto \letbind(\rnd~k,
y.\letbind(g,x.f))[\vec{w}/dom(\Delta)],\] with $y \notin FV(f)$. From lemmas
\ref{thm:predpres} and \ref{thm:substitution} we have that
${f[\vec{w}/dom(\Delta)][r/x] \in \mathcal{VR}^m_{M_{u'}\tau}}$ for some $r \in
\R$ and $m \in \mathbb{N}$. From the induction hypothesis on $v$ we can
concludethat, for any $r' \in \R$, {$g[\vec{w}/dom(\Delta)][r'/y] \in
\mathcal{VR}^n_{u^*}\sigma$} for some $u^* \le u$ and $n \in \mathbb{N}$, and
it therefore follows from lemma \ref{thm:subsumption} that $\letbind(\rnd \ k,
y.\letbind(g,x.f))[\vec{w}/dom(\Delta)] \in
\mathcal{VR}^{(m+1)+(n+1)}_{M_{s*u+u'}\tau}$. The conclusion follows from
lemmas \ref{thm:predpres} and \ref{thm:substitution}. \end{proof}

The following lemma clearly follows by definition of the reducibility
predicate.
\begin{lemma} 
If $e \in \mathcal{R}_\tau$ then there exists a $v \in
CV(\tau)$ such that $e \mapsto^* v$. \label{thm:SN2} 
\end{lemma}
The proof of termination (\Cref{thm:SN}) then follows from \Cref{thm:SN2}
and \Cref{thm:SN1}.
\fi

\section{Denotational Semantics and Error Soundness}\label{sec:metatheory}

In this section, we show two central guarantees of \Lang: bounded
sensitivity and bounded error. 

\subsection{Categorical Preliminaries}\label{subsec:cat}

We provide a denotational semantics for our language based on  
the categorical semantics of Fuzz, due to \citet{DBLP:conf/popl/AmorimGHKC17}. 
Our language has many similarities to Fuzz,
with some key differences needed for our application---most notably, our
language does not have recursive types and non-termination, but it does have a
novel graded monad which we will soon discuss.
We emphasize that we use category theory as a concise language for defining our
semantics---we are ultimately interested in a specific, concrete interpretation
of our language. The general categorical semantics of Fuzz-like languages has
been studied in prior work~\citep{DBLP:conf/icfp/GaboardiKOBU16}.

\paragraph{Basic concepts.}
To begin, we quickly review some basic concepts from category theory; the
interested reader should consult a textbook for a more gentle
introduction~\citep{AwodeyBook,Leinster_2014}. We will introduce more
specialized concepts as we go along. A \emph{category} $\mathbf{C}$ consists of
a collection $Ob$ of objects, and a collection of morphisms $Hom_\mathbf{C}(A,
B)$ for every pair of objects $A, B \in Ob_\mathbf{C}$. For every pair of
morphisms $f \in Hom_\mathbf{C}(A, B)$ and $g \in Hom_\mathbf{C}(B, C)$, the
\emph{composition} $g \circ f$ is defined to be a morphism in $Hom_\mathbf{C}(A,
C)$. There is an \emph{identity} morphism $id_A \in Hom_\mathbf{C}(A, A)$
corresponding to object $A$; this morphism acts as the identity under
composition: $f \circ id = id \circ f = f$.

A \emph{functor} $F$ from category $\mathbf{C}$ to category $\mathbf{D}$
consists of a function on objects $F : Ob_\mathbf{C} \to Ob_\mathbf{D}$, and a
function on morphisms $F : Hom_\mathbf{C}(A, B) \to Hom_\mathbf{D}(A, B)$. The
mapping on morphisms should preserve identities and composition: $F(id_A) =
id_A$, and $F(g \circ f) = F(g) \circ F(f)$. Finally, a \emph{natural
transformation} $\alpha$ from a functor $F : \mathbf{C} \to \mathbf{D}$ to a
functor $G : \mathbf{C} \to \mathbf{D}$ consists of a family of morphisms
$\alpha_A \in Hom_\mathbf{D}(F(A), G(A))$, one per object $A \in Ob_\mathbf{C}$,
that commutes with the functor $F$ and $G$ applied to any morphism: for every $f
\in Hom_\mathbf{C}(A, B)$, we have $\alpha_B \circ F(f) = G(f) \circ \alpha_A$.

\paragraph{The category $\Met$.}

Our type system is designed to bound the distance between various kinds of
program outputs. Intuitively, types should be interpreted as \emph{metric
spaces}, which are sets equipped with a distance function satisfying several
standard axioms. \citet{DBLP:conf/popl/AmorimGHKC17} identified the following
slight generalization of metric spaces as a suitable category to interpret Fuzz.

\begin{definition} \label{def:met-cat}
  An \emph{extended pseudo-metric space} $(A, d_A)$ consists of a \emph{carrier}
  set $A$ and a \emph{distance} $d_A : A \times A \to \NNR \cup \{ \infty \}$
  satisfying (i) reflexivity: $d(a, a) = 0$; (ii) symmetry: $d(a, b) = d(b, a)$;
  and (iii) triangle inequality: $d(a, c) \leq d(a, b) + d(b, c)$ for all $a, b,
  c, \in A$. We write $|A|$ for the carrier set.

  A \emph{non-expansive map} $f : (A, d_a) \to (B, d_B)$ between extended
  pseudo-metric spaces consists of a set-map $f : A \to B$ such that $d_B(f(a),
  f(a')) \leq d_A(a, a')$. The identity function is a non-expansive map, and
  non-expansive maps are closed under composition. Therefore, extended
  pseudo-metric spaces and non-expansive maps form a category $\Met$.
\end{definition}

Extended pseudo-metric spaces differ from standard metric spaces in two
respects. First, their distance functions can assign infinite distances
(\emph{extended} real numbers). Second, their distance functions are only
\emph{pseudo}-metrics because they can assign distance zero to pairs of distinct
points. Since we will only be concerned with extended pseudo-metric spaces, we
will refer to them as metric spaces for short.

The category $\Met$ supports several constructions that are useful for
interpreting linear type systems. First, there are products and coproducts on
$\Met$. The Cartesian product $(A, d_A) \times (B, d_B)$ has carrier $A \times
B$ and distance given by the max: $d_{A \times B} ((a, b), (a', b')) =
\text{max}(d_A(a, a'), d_B(b, b'))$. The tensor product $(A, d_A) \otimes (B,
d_B)$ also has carrier $A \times B$, but with distance given by the sum: $d_{A
\otimes B} ((a, b), (a', b')) = d_A(a, a') + d_B(b, b')$. Both products are
useful for modeling natural metrics on pairs and tuples.  The category $\Met$
also has coproducts $(A, d_A) + (B, d_B)$, where the carrier is disjoint union
$A \uplus B$ and the metric $d_{A + B}$ assigns distance $\infty$ to pairs of
elements in different injections, and distance $d_A$ or $d_B$ to pairs of
elements in $A$ or $B$, respectively.

Second, non-expansive functions can be modeled in $\Met$. The function space
$(A, d_A) \lin (B, d_B)$ has carrier set $\{ f : A \to B \mid
f~\text{non-expansive} \}$ and distance given by the supremum norm: $d_{A
\lin B}(f, g) = \text{sup}_{a \in A} d_B(f(a), g(a))$. Moreover, the
functor $(- \otimes B)$ is left-adjoint to the functor $(B \lin -)$, so
maps $f : A \otimes B \to C$ can be curried to $\lambda(f) : A \to (B \lin
C)$, and uncurried. These constructions, plus a few additional pieces of data,
make $(\Met, I, \otimes, \lin)$ a \emph{symmetric monoidal closed category}
(SMCC), where the unit object $I$ is the metric space with a single element.

\paragraph{A graded comonad on $\Met$.}

Languages like Fuzz are based on \emph{bounded linear
logic}~\citep{DBLP:journals/tcs/GirardSS92}, where the exponential type $!A$ is
refined into a family of bounded exponential types $!_s A$ where $s$ is drawn
from a pre-ordered semiring $\mathcal{S}$. The grade $s$ can be used to track
more fine-grained, possibly quantitative aspects of well-typed terms, such as
function sensitivities. These bounded exponential types can be modeled by a
categorical structure called a $\mathcal{S}$-\emph{graded exponential
comonad}~\citep{DBLP:conf/esop/BrunelGMZ14,DBLP:conf/icfp/GaboardiKOBU16}. Given
any metric space $(A, d_A)$ and non-negative number $r$, there is an evident
operation that scales the metric by $r$: $(A, r \cdot d_A)$. This operation can
be extended to a graded comonad.

\begin{definition} \label{def:scale-comonad}
  Let the pre-ordered semiring $\mathcal{S}$ be the extended non-negative real
  numbers $\NNR \cup \{ \infty \}$ with the usual order, addition, and
  multiplication; $0 \cdot \infty$ and $\infty \cdot 0$ are defined to be $0$.
  We define functors $\{ D_s : \Met \to \Met \mid s \in \mathcal{S} \}$ such
  that $D_s : \Met \to \Met$ takes metric spaces $(A, d_A)$ to metric spaces
  $(A, s \cdot d_A)$, and non-expansive maps $f : A \to B$ to $D_s f : D_s A \to
  D_s B$, with the same underlying map.

  \ifshort
\else
  We also define the following associated natural transformations:
  \begin{itemize}
    \item For $s, t \in \mathcal{S}$ and $s \leq t$, the map $(s \leq t)_A : D_t
      A \to D_s A$ is the identity; note the direction.
    \item The map $m_{s, I} : I \to D_s I$ is the identity map on the singleton
      metric space.
    \item The map $m_{s, A, B} : D_s A \otimes D_s B \to D_s (A \otimes B)$ is the
      identity map on the underlying set.
    \item The map $w_A : D_0 A \to I$ maps all elements to the singleton.
    \item The map $c_{s, t, A} : D_{s + t} A \to D_s A \otimes D_t A$ is the
      diagonal map taking $a$ to $(a, a)$.
    \item The map $\epsilon_A : D_1 A \to A$ is the identity.
    \item The map $\delta_{s, t, A} : D_{s \cdot t} A \to D_s (D_t A)$ is the identity.
  \end{itemize}
\fi
\end{definition}

\ifshort
\else
These maps are all non-expansive and it can be shown that they satisfy the
diagrams~\citep{DBLP:conf/icfp/GaboardiKOBU16} defining a $\mathcal{S}$-graded
exponential comonad, but we will not need these abstract equalities for our
purposes.
\fi

\subsection{A Graded Monad on $\Met$}\label{subsec:monad}

The categorical structures we have seen so far are enough to interpret the
non-monadic fragment of our language, which is essentially the core of the Fuzz
language~\citep{DBLP:conf/popl/AmorimGHKC17}. As proposed
by~\citet{DBLP:conf/icfp/GaboardiKOBU16}, this core language can model effectful
computations using a graded monadic type, which can be modeled categorically by
(i) a \emph{graded strong monad}, and (ii) a \emph{distributive law} modeling
the interaction of the graded comonad and the graded monad.

\paragraph{The neighborhood monad.}

Recall the intuition behind our system: closed programs $e$ of type $M_{\rnderr} \num$ are
computations producing outputs in $\num$ that may perform rounding operations. The
index $\rnderr$ should bound the distance between the output under the \emph{ideal}
semantics, where rounding is the identity, and the \emph{floating-point (FP)}
semantics, where rounding maps a real number to a representable floating-point
number following a prescribed rounding procedure. Accordingly, the
interpretation of the graded monad should track \emph{pairs} of
values---the ideal value, and the FP value.

This perspective points towards the following graded monad on $\Met$, which we
call the \emph{neighborhood monad}. While the definition appears quite natural
mathematically, we are not aware of this graded monad appearing in prior work.

\begin{definition} \label{def:nhd-monad}
  Let the pre-ordered monoid $\mathcal{R}$ be the extended non-negative real
  numbers $\NNR \cup \{ \infty \}$ with the usual order and addition. The
  \emph{neighborhood monad} is defined by the functors $\{ T_r : \Met \to \Met
  \mid r \in \mathcal{R} \}$ and associated natural transformations as follows:
  \begin{itemize}
    \item The functor $T_r : \Met \to \Met$ takes a metric space
      $M$ to a metric space with underlying set:
      \[
        |T_r M| \triangleq \{ (x, y) \in M \times M \mid d_M(x, y) \leq r \}
      \]
      and the metric is:
      $
        d_{T_r M} ( (x, y), (x', y') ) \triangleq d_M (x, x'). 
      $
    \item The functor $T_r$ takes a non-expansive function $f : A \to B$ to
      $T_r f : T_r A \to T_r B$ with
      \[
        (T_r f)( (x, y) ) \triangleq (f(x), f(y))
      \]
    \item For $r, q \in \mathcal{R}$ and $q \leq r$, the map $(q \leq r)_A : T_q
      A \to T_r A$ is the identity.
    \item The unit map $\eta_A : A \to T_0 A$ is defined via:
      $
        \eta_A(x) \triangleq (x, x).
      $
    \item The graded multiplication map $\mu_{q, r, A} : T_q (T_r A) \to T_{r + q}
      A$ is defined via:
      \[
        \mu_{q, r, A} ( (x, y), (x', y') ) \triangleq (x, y').
      \]
  \end{itemize}
\end{definition}

The definitions of $T_r$ are evidently functors. The associated maps are natural
transformations, and define a graded
monad~\citep{DBLP:conf/popl/Katsumata14,DBLP:conf/fossacs/FujiiKM16}. %
\ifshort
\else
\begin{lemma}\label{lem:monad-nat}
  Let $q, r \in \mathcal{R}$. For any metric space $A$, the maps $(q \leq r)_A$,
  $\eta_A$, and $\mu_{q, r, A}$ are non-expansive maps and natural in $A$.
\end{lemma}
\begin{proof}
  Non-expansiveness and naturality for the subeffecting maps $(q \leq r)_A : T_q
  A \to T_r A$ and the unit maps $\eta_A : A \to T_0 A$ are straightforward. We
  describe the checks for the multiplication map $\mu_{q, r, A}$. 
  
  First, we check that the multiplication map has the claimed domain and
  codomain.  Note that $d_A(x, x') = d_{T_r A} ( (x, y), (x', y') ) \leq q$
  because of the definition of $T_q$, and $d_A(x', y') \leq r$ because of the
  definition of $T_r$, so via the triangle inequality we have $d_A(x, y') \leq r
  + q$ as claimed.

  Second, we check non-expansiveness. Let $( (x, y), (x', y')$ and $( (w, z),
  (w', z') )$ be two elements of $T_q (T_r A)$. Then:
  \begin{align}
    d_{T_{r + q} A}( \mu( (x, y), (x', y') ) , \mu( (w, z), (w', z') ) )
    &= d_{T_{r + q} A}( (x, y'), (w, z') ) \tag{def. $\mu$} \\
    &= d_A(x, w) \tag{def. $d_{T_{r + q} A}$} \\
    &= d_{T_r A} ( (x, y), (w, z) ) \tag{def. $d_{T_r} A$} \\
    &= d_{T_q (T_r A)} ( ((x, y), (x', y')), ((w, z), (w', z')) ) \tag{def. $d_{T_q (T_r A)}$}
  \end{align}
  
  Finally, we can check naturality. Let $f : A \to B$ be any non-expansive map.
  By unfolding definitions, it is straightforward to see that $\mu_{q, r, B}
  \circ T_q (T_r f) = T_{r + q} f \circ \mu_{q, r, A}$.
\end{proof}

\begin{lemma} \label{lem:monad}
  The functors $T_r$ and associated maps form a $\mathcal{R}$-graded monad on
  $\Met$.
\end{lemma}
\begin{proof}
  Establishing this fact requires checking the following two diagrams:
\[\begin{tikzcd}
	{T_r A} && {T_0(T_rA)} \\
	\\
	{T_r (T_0 A)} && {T_r A}
	\arrow["{T_r \eta_A}"', from=1-1, to=3-1]
	\arrow["{\eta_{T_r A}}", from=1-1, to=1-3]
	\arrow["{\mu_{0, r, A}}", from=1-3, to=3-3]
	\arrow["{\mu_{r, 0, A}}"', from=3-1, to=3-3]
	\arrow[shift right=1, no head, from=1-1, to=3-3]
	\arrow[shift left=1, no head, from=1-1, to=3-3]
\end{tikzcd}
\qquad\qquad
%
\begin{tikzcd}
	{T_p(T_q(T_rA))} && {T_p(T_{q + r}A)} \\
	\\
	{T_{p + q} (T_r A)} && {T_{p + q + r} A}
	\arrow["{\mu_{p, q, T_r A}}"', from=1-1, to=3-1]
	\arrow["{\mu_{p + q, r, A}}"', from=3-1, to=3-3]
	\arrow["{T_p \mu_{q, r, A}}", from=1-1, to=1-3]
	\arrow["{\mu_{p, q + r, A}}", from=1-3, to=3-3]
\end{tikzcd}\]
Both diagrams follow by unfolding definitions.
\end{proof}

\paragraph{Graded monad strengths.}
As proposed by~\citet{DBLP:journals/iandc/Moggi91}, monads that model
computational effects should be \emph{strong monads}; roughly, they should
behave well with respect to products.

\begin{definition} \label{def:nhd-st}
  Let $r \in \mathcal{R}$. We define non-expansive maps
  $st_{r, A, B} : A \otimes T_r B \to T_r (A \otimes B)$ via
  \begin{align*}
    st_{r, A, B}(a, (b, b')) &\triangleq ((a, b), (a, b'))
  \end{align*}
  Moreover, these maps are natural in $A$ and $B$.
\end{definition}

We can check non-expansiveness. For $(a, (b, b'))$ and $(c, (d, d'))$ in $A
\otimes T_r B$, we have:
\begin{align*}
  d_{T_r (A \otimes B)} (st (a, (b, b')), st (c, (d, d')))
  &= d_{T_r (A \otimes B)} (((a, b), (a, b')), ((c, d), (c, d')))
  \tag{def.  $st$} \\
  &= d_{A \otimes B} ((a, b), (c, d))
  \tag{def. $d_{T_r (A \otimes B)}$}
  \\
  &= d_A (a, c) + d_B (b, d)
  \tag{def. $d_{A \otimes B}$} \\
  &= d_A (a, c) + d_{T_r B} ((b, b'), (d, d'))
  \tag{def. $d_{T_r B}$} \\
  &= d_{A \otimes T_r B}  ((a, (b, b')), (c, (d, d')))
  \tag{def. $d_{A \otimes T_r B}$}
\end{align*}

Furthermore, it is possible to show that these maps satisfy the commutative
diagrams needed to make $T_r$ a graded strong
monad~\citep{DBLP:conf/popl/Katsumata14,DBLP:conf/fossacs/FujiiKM16}, though we
will not need this fact for our development. Finally, we have shown that the
neighborhood monad is strong with respect to the tensor product $\otimes$; in
fact, the neighborhood monad is also strong with respect to the Cartesian
product $\times$.

\paragraph{Graded distributive law.}

\citet{DBLP:conf/icfp/GaboardiKOBU16} showed that languages supporting graded
coeffects and graded effects can be modeled with a graded comonad, a graded
monad, and a graded distributive law. In our model, we have the following family
of maps.

\begin{lemma} \label{lem:distr}
  Let $s \in \mathcal{S}$ and $r \in \mathcal{R}$ be grades, and let $A$ be a
  metric space. Then identity map on the carrier set $|A| \times |A|$ is a
  non-expansive map
  \[
    \lambda_{s, r, A} : D_s (T_r A) \to T_{s \cdot r} (D_s A)
  \]
  Moreover, these maps are natural in $A$.
\end{lemma}

\begin{proof}
  We first check the domain and codomain. Let $x, y \in A$ be such that $(x, y)$
  is in the domain $D_s(T_r A)$ of the map. Thus $(x, y)$ must also be in $T_r
  A$, and satisfy $d_A(x, y) \leq r$ by definition of $T_r$.  To show that this
  element is also in the range, we need to show that $d_{D_s A} (x, y) \leq s
  \cdot r$, but this holds by definition of $D_s$.
  We can also check that this map is non-expansive:
  \begin{align*}
    d_{T_{s \cdot r} (D_s A)} ((x, y), (x', y'))
  &\triangleq d_{D_s A} (x, x')
  \tag{def. $T_{s \cdot r}$} \\
  &\triangleq s \cdot d_A(x, x')
  \tag{def. $D_s$} \\
  &\triangleq s \cdot d_{T_r A}((x, y), (x', y'))
  \tag{def. $T_r$} \\
  &\triangleq d_{D_s(T_r A)}((x, y), (x', y'))
  \tag{def. $D_s$}
  \end{align*}
  Since $\lambda_{s, r, A}$ is the identity map on the underlying set $|A|
  \times |A|$, it is evidently natural in $A$.
\end{proof}

It is similarly straightforward to show that the maps $\lambda_{s, r, A}$ form a
graded distributive law in the sense of~\citet{DBLP:conf/icfp/GaboardiKOBU16}:
for $s \leq s'$ and $r \leq r'$ the identity map $T_{s \cdot r}(D_s A) \to T_{s'
\cdot r'}(D_{s'} A)$ is also natural in $A$, and the four diagrams required for
a graded distributive law all commute~\citep[Fig.
8]{DBLP:conf/icfp/GaboardiKOBU16}, but since we do not rely on these properties
we will omit these details.
\fi

\subsection{Interpreting the Language}\label{subsec:interp}

We are now ready to interpret our language in $\Met$.

\paragraph{Interpreting types.}

We interpret each type $\tau$ as a metric space $\denot{\tau}$, using
constructions in $\Met$.

\begin{definition}\label{def:interp-ty}
  Define the type interpretation by induction on the type syntax:
  \begin{mathpar}
    \denot{\unit} \triangleq I = (\{ \star \}, 0) \and
    \denot{\num} \triangleq (R, d_R) \and
    \denot{A \otimes B} \triangleq \denot{A} \otimes \denot{B} \and
    \denot{A \times B} \triangleq \denot{A} \times \denot{B} \and
    \denot{A + B} \triangleq \denot{A} + \denot{B} \and
    \denot{A \lin B} \triangleq \denot{A} \lin \denot{B} \and
    \denot{\bang{s} A} \triangleq D_s \denot{A} \and
    \denot{M_r A} \triangleq T_r \denot{A}
  \end{mathpar}
\end{definition}
We do not fix the interpretation of the base type $\num$: $(R, d_R)$ can be any
metric space.

\paragraph{Interpreting judgments.}

We will interpret each typing derivation showing a typing judgment $\Gamma
\vdash e : \tau$ as a morphism in $\Met$ from the metric space
$\denot{\Gamma}$ to the metric space $\denot{\tau}$. Since all morphisms in this
category are non-expansive, this will show (a version of) metric preservation.
We first define the metric space $\denot{\Gamma}$:
\begin{align*}
  \denot{\cdot} \triangleq I = (\{ \star \}, 0) 
	\qquad  \qquad
  \denot{\Gamma, x :_s \tau} \triangleq \denot{\Gamma} \otimes D_s \denot{\tau}
\end{align*}
Given any binding $x :_r \tau \in \Gamma$, there is a non-expansive map from
$\denot{\Gamma}$ to $\denot{\tau}$ projecting out the $x$-th position; we
sometimes use notation that treats an element $\gamma \in \denot{\Gamma}$ as a
function, so that $\gamma(x) \in \denot{\tau}$. %
\ifshort
\else
Formally, projections are defined via the weakening maps $(0 \leq s)_A ; w_A :
D_s A \to I$ and the unitors.

We begin with two simple lemmas about context addition ($\Gamma + \Delta$) and
context scaling $s \cdot \Gamma$.

\begin{lemma} \label{lem:ctx-contr}
  Let $\Gamma$ and $\Delta$ such that $\Gamma + \Delta$ is defined. Then there
  is a non-expansive map $c_{\Gamma, \Delta} : \denot{\Gamma + \Delta} \to
  \denot{\Gamma} \otimes \denot{\Delta}$ given by:
  \[
    c_{\Gamma, \Delta}(\gamma) \triangleq (\gamma_\Gamma, \gamma_\Delta)
  \]
  where $\gamma_\Gamma$ and $\gamma_\Delta$ project out the positions in
  $dom(\Gamma)$ and $dom(\Delta)$, respectively.
\end{lemma}

\begin{lemma} \label{lem:ctx-scale}
  Let $\Gamma$ be a context and $s \in \mathcal{S}$ be a sensitivity. Then the
  identity function is a non-expansive map from $\denot{s \cdot \Gamma} \to D_s
  \denot{\Gamma}$.
\end{lemma}
\fi

We are now ready to define our interpretation of typing judgments. Our
definition is parametric in the interpretation of three things: the numeric type $\denot{\num} =
(R, d_R)$, the rounding operation $\rho$, and the operations in the signature $\Sigma$.

\begin{definition} \label{def:interp-prog}
  Fix $\rho : R \to R$ to be a (set) function such that for every $r \in R$ we
  have $d_R(r, \rho(r)) \leq \rnderr$, and for every operation $\{ \mathbf{op} :
  \sigma \lin \tau \} \in \Sigma$ in the signature fix an interpretation
  $\denot{\op} : \denot{\sigma} \to \denot{\tau}$ such that for every closed
  value $\cdot \vdash v : \sigma$, we have $\denot{\op}(\denot{v}) =
  \denot{op(v)}$.

  Then we can interpret each well-typed program
  $\Gamma \vdash e : \tau$ as a non-expansive map $\denot{\Gamma \vdash e :
  \tau} : \denot{\Gamma} \to \denot{\tau}$, by induction on the typing
  derivation, via case analysis on the last rule.
\end{definition}

\ifshort
\else
We detail the cases here. We write our maps in diagrammatic order. To reduce
notation, we sometimes omit indices on maps and we elide the bookkeeping
morphisms from the SMCC structure (the unitors $\lambda_A : I \otimes A \to A$
and $\rho_A : A \otimes I \to I$; the associators $\alpha_{A, B, C} : (A \otimes
B) \otimes C \to A \otimes (B \otimes C)$, and the symmetries $\sigma_{A, B} : A
\otimes B \to B \otimes A$).
\begin{description}
  \item[\textsc{Const}.] Define $\denot{\Gamma \vdash k : \num} :
    \denot{\Gamma} \to \denot{\num}$ to be the constant function returning
    $k \in R$.
  \item[\textsc{Ret}.] Let $f = \denot{\Gamma \vdash v : \tau}$. Define
    $\denot{\Gamma \vdash \ret v : M_0 \tau}$ to be $f ;
    \eta_{\denot{\tau}}$.
  \item[\textsc{Subsumption}.] Let $f = \denot{\Gamma \vdash e : M_r
    \tau}$. Define $\denot{\Gamma \vdash e : M_{r'} \tau}$ to be $f ;
    (r \leq r')_{\denot{\tau}}$.
  \item[\textsc{Round}.] Letting $f = \denot{\Gamma \vdash k : \num}$, we can
    define
    \[
      \denot{\Gamma \vdash \rnd k : M_{\rnderr} \num}
      \triangleq f ; \langle id, \rho \rangle
    \]
    Explicitly, the second map takes $r \in R$ to the pair $(r, \rho(r))$. The
    output is in $\denot{M_\rnderr \num}$ by our assumption of the rounding
    function $\rho$, and the function is non-expansive by the definition of the
    metric on $\denot{M_{\rnderr} \num}$.
  \item[\textsc{Let-Bind}.] Let $f = \denot{\Gamma \vdash e : M_r \sigma}$
    and $g = \denot{\Delta, x:_s \sigma \vdash e' : M_q \tau}$. We need to
    combine these ingredients to define a morphism from $\denot{s \cdot \Gamma +
    \Delta}$ to $\denot{M_{s \cdot r + q} \tau}$. First, we can apply
    the comonad to $f$ and then compose with the distributive law to get:
    \[
      D_s f ; \lambda_{s, r, \denot{\sigma}}
      : D_s \denot{\Gamma} \to T_{s \cdot r} D_s \denot{\sigma} .
    \]
    Repeatedly pre-composing with the map $m_{s, A, B} : D_s A \otimes D_s B \to
    D_s (A \otimes B)$, we get a map from $\denot{s \cdot \Gamma} \to T_{s \cdot
    r} D_s \denot{\sigma}$. Composing in parallel with
    $id_{\denot{\Theta}}$ and post-composing with the strength, we have:
    \[
      ((m ; D_s f ; \lambda_{s, r, \sigma}) \otimes id_{\denot{\Theta}})
      ; \otimes st_{s \cdot r, \denot{\Theta}, \denot{\sigma}}
      : \denot{\Theta} \otimes \denot{s \cdot \Gamma}
      \to T_{s \cdot r} (\denot{\Theta} \otimes D_s \denot{\sigma})
    \]
    Next, applying the functor $T_{s \cdot r}$ to $g$ and then
    post-composing with the multiplication $\mu_{s \cdot r, q,
    \denot{\tau}}$, we have:
    \[
      T_{s \cdot r} g ; \mu_{s \cdot r, q, \denot{\sigma}}
      : T_{s \cdot r} (\denot{\Theta} \otimes D_s \denot{\sigma})
      \to T_{s \cdot r + q} \denot{\tau} ,
    \]
    which can be composed with the previous map to get a map from
    $\denot{\Theta} \otimes \denot{s \cdot \Gamma}$ to $T_{s \cdot r +
    q} \denot {\tau} = \denot{M_{s \cdot r + q} \tau}$.
    Pre-composing with $c_{s \cdot \denot{\Gamma}, \denot{\Theta}}$ and a swap
    gives the desired map from $\denot{s \cdot \Gamma + \Theta}$ to $\denot{M_{s
    \cdot r + q} \tau}$.
  \item [\textsc{Op}.] By assumption, we have an interpretation $\denot{op}$ for
    every operation in the signature $\Sigma$. We interpret $\denot{\Gamma
    \vdash \op(v) : \tau}$ as the composition $\denot{\Gamma \vdash v : \sigma}
    ; \denot{\op}$.
  \item[\textsc{Var}.] We define $\denot{\Gamma \vdash x : \tau}$ to be the map
    that maps $\denot{\Gamma}$ to the $x$-th component $\denot{\tau}$. All other
    components are mapped to $I$ and then removed with the unitor.
  \item[\textsc{Abstraction}.] Let $f = \denot{\Gamma, x :_1 \sigma \vdash e :
    \tau} : \denot{\Gamma} \otimes D_1 \denot{\sigma} \to \denot{\tau}$. Now,
    $D_1 \denot{\sigma} = \denot{\sigma}$ in our model. Since $\Met$ is an SMCC,
    we have a map $\lambda(f) : \denot{\Gamma} \to (\denot{\sigma} \lin
    \denot{\tau})$.
  \item[\textsc{Application}.] Since $\Met$ is an SMCC, there is a morphism $ev
    : (A \lin B) \otimes A \to B$. Let $f = \denot{\Gamma \vdash v : \sigma
    \lin \tau}$ and $g = \denot{\Theta \vdash w : \sigma}$. Then we can
    define:
    \[
      c_{\denot{\Gamma}, \denot{\Theta}} ; (f \otimes g) ; ev
      : \denot{\Gamma + \Theta} \to \denot{\tau}
    \]
\item[\textsc{Unit}.] We let $\denot{\Gamma \vdash \langle\rangle :
  \unit}$ be the map that sends all points in $\Gamma$ to $\star \in
  \denot{\unit}$.
  \item[$\times$~\textsc{Intro}.] Letting $f$ and $g$ be the denotations of the
    premises, take $\langle f, g \rangle : \denot{\Gamma} \to \denot{\sigma
    \times \tau}$.
  \item[$\times$~\textsc{Elim}.] Let $f$ be the denotation of the premise, and
    post-compose by the projection $\pi_i$ to get a map $\denot{\Gamma} \to
    \denot{\tau_i}$.
  \item[$\otimes$~\textsc{Intro}.] Let $f$ and $g$ be the denotations of the
    premises, and define:
    \[
      c_{\denot{\Gamma}, \denot{\Theta}} ; (f \otimes g)
      : \denot{\Gamma + \Theta} \to \denot{\sigma \otimes \tau}
    \]
  \item[$\otimes$~\textsc{Elim}.] Let $f = \denot{\Gamma \vdash v : \sigma
    \otimes \tau} : \denot{\Gamma} \to \denot{\sigma} \otimes \denot{\tau}$ and
    $g = \denot{\Theta, x :_s \sigma, y :_s \tau \vdash e : \rho} :
    \denot{\Theta} \otimes D_s \denot{\sigma} \otimes D_s \denot{\tau} \to
    \denot{\rho}$. Applying the functor $D_s$ to $f$ and pre-composing with $m$,
    we get
    \[
      m ; D_s f : \denot{s \cdot \Gamma} \to D_s (\denot{\sigma} \otimes \denot{\tau})
    \]
    Since the map $m$ is the identity in our model, we can post-compose by its
    inverse to get:
    \[
      m ; D_s f ; m^{-1} : \denot{s \cdot \Gamma} \to D_s (\denot{\sigma}) \otimes D_s (\denot{\tau})
    \]
    Composing in parallel with $id_{\denot{\Theta}}$, we get:
    \[
      id_{\denot{\Theta}} \otimes (m ; D_s f ; m^{-1})
      : \denot{\Theta} \otimes \denot{s \cdot \Gamma}
      \to \denot{\Theta} \otimes D_s \denot{\sigma} \otimes D_s \denot{\tau}
    \]
    Post-composing with $g$ and pre-composing with $c_{\denot{s \cdot \Gamma},
    \denot{\Theta}}$ completes the definition.
  \item[$+$~\textsc{Intro L}.] Let $f$ be the denotation of the premise, and
    take $f ; \iota_1$ where $\iota_2$ is the first injection into the coproduct.
  \item[$+$~\textsc{Intro R}.] Let $f$ be the denotation of the premise, and
    take $f ; \iota_2$ where $\iota_2$ is the second injection into the
    coproduct.
  \item[$+$~\textsc{Elim}.] We need a few facts about our model. First, there is
    an isomorphism $dist_{D, s} : D_s (A + B) \cong D_s A + D_s B$ when $s$ is
    \emph{strictly} greater than zero.  Second, there is a map $dist_\times : A
    \times (B + C) \to A \times B + A \times C$, which pushes the first
    component into the disjoint union. This map is non-expansive, and in fact
    $\Met$ is a distributive category.

    Let $f = \denot{\Gamma \vdash v : \sigma + \tau}$ and $g_i = \denot{\Theta,
    x_i :_s \sigma \vdash e_i : \rho}$ for $i = 1, 2$. Since $s$ is
    strictly greater than zero, $dist_{D, s}$ is an isomorphism. By
    using the functor $D_{s}$ on $f$, composing in parallel with
    $id_{\denot{\Theta}}$ and distributing, we have:
    \[
      (id_{\denot{\Theta}} \otimes (m ; D_{s} f ; dist_{D, s})) ; dist_\times
      : \denot{\Theta} \otimes \denot{(s}
      \to \denot{\Theta} \otimes D_{s} \denot{\sigma}
      + \denot{\Theta} \otimes D_{s} \denot{\tau}
    \]
    By post-composing with the pairing map $[g_1, g_2]$ from the coproduct, and
    pre-composing with $c_{\denot{s \cdot \Gamma}, \denot{\Theta}}$ and a swap,
    we get a map $\denot{s \cdot \Gamma + \Theta} \to \denot{\rho}$ as desired.
  \item[$\bang{s}$~\textsc{Intro}.] Let $f$ be the denotation of the premise, and
    take $m ; D_s f$.
  \item[$\bang{s}$~\textsc{Elim}.] Let $f = \denot{\Gamma \vdash v : \bang{s}
    \sigma}$ and $g = \denot{\Theta, x :_{m \cdot n} \sigma \vdash e : \tau}$.
    We have:
    \[
      m ; D_m f ; \delta_{m, n, \denot{\sigma}}^{-1}
      : \denot{m \cdot \Gamma} \to D_{m \cdot n} \denot{\sigma}
    \]
    Here we use that $\delta_{m, n, \denot{\sigma}}$ is an
    isomorphism in our model. By composing in parallel with
    $id_{\denot{\Theta}}$, we can then post-compose by $g$. Pre-composing with
    $c_{\denot{s \cdot \Gamma}, \denot{\Theta}}$ gives a map $\denot{s
    \cdot \Gamma + \Theta} \to \denot{\tau}$, as desired.
  \item[\textsc{Let}.] This case is essentially the same as the other
    elimination cases.
\end{description}
\fi

\paragraph{Soundness of operational semantics.}
Now, we can show that the operational semantics from \Cref{sec:language} is
sound with respect to the metric space semantics: stepping a well-typed term
does not change its denotational semantics.

\ifshort
\else
We first need a weakening lemma.

\begin{lemma}[Weakening] \label{lem:weak-sem}
  Let $\Gamma, \Gamma' \vdash e : \tau$ be a well-typed term. Then for any
  context, there is a derivation of $\Gamma, \Delta, \Gamma' \vdash e : \tau$
  with semantics $\denot{\Gamma, \Gamma' \vdash e : \tau} \circ \pi$, where $\pi
  : \denot{\Gamma, \Delta, \Gamma'} \to \denot{\Gamma, \Gamma'}$ projects the
  components in $\Gamma$ and $\Gamma'$.
\end{lemma}

\begin{proof}
  By induction on the typing derivation of $\Gamma, \Gamma' \vdash e: \tau$.
\end{proof}

Next, we show that the subsumption rule is admissible.

\begin{lemma}[Subsumption] \label{lem:subsump}
  Let $\Gamma \vdash e : M_r \tau$ be a well-typed program of monadic type,
  where the typing derivation concludes with the subsumption rule. Then either
  $e$ is of the form $\ret v$ or $\rnd k$, or there is a
  derivation of $\Gamma \vdash e : M_r \tau$ with the same semantics that
  does not conclude with the subsumption rule.
\end{lemma}
\begin{proof}
  By straightforward induction on the typing derivation, using the fact that
  subsumption is transitive, and the semantics of the subsumption rule leaves
  the semantics of the premise unchanged since the subsumption map $(r \leq
  s)_A$ is the identity function.
\end{proof}
\fi

\begin{lemma}[Substitution] \label{lem:subst-sem}
  Let $\Gamma, \Delta, \Gamma' \vdash e : \tau$ be a well-typed term, and let
  $\vec{v} : \Delta$ be a well-typed substitution of closed values, i.e., we
  have derivations $\cdot \vdash v_x : \Delta(x)$. Then there is a derivation of
  \[
    \Gamma, \Gamma' \vdash e[\vec{v}/dom(\Delta)] : \tau
  \]
  with semantics
  $
    \denot{\Gamma, \Gamma' \vdash e[\vec{v}/dom(\Delta)] : \tau}
    = (id_{\denot{\Gamma}} \otimes \denot{\cdot \vdash \vec{v} : \Delta} \otimes id_{\denot{\Gamma'}})
    ; \denot{\Gamma, \Delta, \Gamma' \vdash e : \tau} .
  $
\end{lemma}
\ifshort
\else
\begin{proof}
  By induction on the typing derivation of $\Gamma, \Delta, \Gamma' \vdash e :
  \tau$. The base cases Unit and Const are obvious. The other base case Var
  follows by unfolding the definition of the semantics. Most of the rest of the
  cases follow from the substitution lemma for Fuzz~\citep[Lemma
  3.3]{DBLP:conf/popl/AmorimGHKC17}. We show the cases for \textsc{Round},
  \textsc{Return}, and \textsc{Let-Bind}, which differ from Fuzz. We omit the
  bookkeeping morphisms.

  \begin{description}
    \item[Case \textsc{Round}.] Given a derivation $f = \denot{\Gamma, \Delta,
      \Gamma' \vdash w : \num}$, by induction, there is a derivation
      $\denot{\Gamma, \Gamma' \vdash w[\vec{v}/\Delta] : M_{\rnderr} \num} =
      (id_{\denot{\Gamma}} \otimes \denot{\cdot \vdash \vec{v} : \Delta} \otimes
      id_{\denot{\Gamma'}}) ; f$. By applying rule \textsc{Round} and by
      definition of the semantics of this rule, we have a derivation
      \[
        \denot{\Gamma, \Gamma' \vdash (\rnd w)[\vec{v}/\Delta] : M_{\rnderr} \num}
        = (id_{\denot{\Gamma}} \otimes \denot{\cdot \vdash \vec{v} : \Delta} \otimes id_{\denot{\Gamma'}}) ; f ; \langle id, \rho \rangle .
      \]
      We are done since $\denot{\Gamma, \Delta, \Gamma' \vdash \rnd~w :
      M_{\rnderr} \num} = f ; \langle id, \rho \rangle$.
    \item[Case \textsc{Return}.] Same as previous, using the unit of the monad
      $\eta_{\denot{\tau}}$ in place of $\langle id, \rho \rangle$.
    \item[Case \textsc{Let-Bind}.] Suppose that $\Gamma = \Gamma_1, \Delta_1,
      \Gamma_2$ and $\Theta = \Theta_1, \Delta_2, \Theta_2$ such that $\Delta =
      s \cdot \Delta_1 + \Delta_2$. By combining \Cref{lem:ctx-contr} and
      \Cref{lem:ctx-scale}, there is a natural transformation $\sigma : \denot{s
      \cdot \Gamma + \Theta} \to \denot{\Theta} \otimes D_s \denot{\Gamma}$.

      Let $g_1 = \denot{\Gamma_1, \Delta_1, \Gamma_2 \vdash w : M_r \sigma}$ and
      $g_2 = \denot{\Theta_11, \Delta_2, \Theta_2, x :_s \sigma \vdash f : M_{r'} \tau}$. By
      induction, we have:
      \begin{align*}
        \tilde{g_1} &= \denot{\Gamma_1, \Gamma_2 \vdash w[\vec{v}/\Delta] : M_r \sigma}
        = (id_{\denot{\Gamma_1}} \otimes \denot{\cdot \vdash \vec{v} : \Delta}
        \otimes id_{\denot{\Gamma_2}}) ; g_1 \\
        \tilde{g_2} &= \denot{\Theta_1, \Theta_2, x :_s \sigma \vdash f[\vec{v}/\Delta] : M_{r'} \tau}
        = (id_{\denot{\Theta_1}} \otimes \denot{\cdot \vdash \vec{v} : \Delta}
        \otimes id_{\denot{\Theta_2, x :_s \sigma}}) ; g_2
      \end{align*}
      Thus we have a derivation of the judgment $s \cdot (\Gamma_1, \Gamma_2) +
      (\Theta_1, \Theta_2) \vdash \letbind(w, x.f)[\vec{v} / \Delta] : M_{s
      \cdot r + r'} \tau$, and by the definition of the semantics of
      \textsc{Let-Bind}, its semantics is:
      \[
        split
        ;
        (id_{\denot{\Theta_1, \Theta_2}} \otimes (D_s \tilde{g_1} ; \lambda_{s, r, \denot{\sigma}}))
        ;
        st_{\denot{\Theta_1, \Theta_2}, \denot{\sigma}}
        ;
        T_{s \cdot r} \tilde{g_2}
        ;
        \mu_{s \cdot r, r', \denot{\tau}}
      \]
      From here, we can conclude by showing that the first morphisms in
      $\tilde{g_1}$ and $\tilde{g_2}$ can be pulled out to the front. For
      instance,
      \[
        D_s \tilde{g_1} = (id_{\denot{s \cdot \Gamma_1}}
          \otimes D_s \denot{\cdot \vdash \vec{v} : \Delta}
          \otimes id_{\denot{s \cdot \Gamma_2}})
          ; D_s g_1
      \]
      by functoriality. By naturality of $split$, the first morphism can be
      pulled out in front of $split$.

      Similarly, for $\tilde{g_2}$, we have:
      \begin{multline*}
        st_{\denot{\Theta_1, \Theta_2}, \denot{\sigma}} ;
        T_{s \cdot r} (id_{\denot{\Theta_1}} \otimes \denot{\cdot \vdash \vec{v} : \Delta} \otimes id_{\denot{\Theta_2, x :_s \sigma}})
        \\
        = (id_{\denot{\Theta_1}} \otimes \denot{\cdot \vdash \vec{v} : \Delta}
        \otimes id_{\denot{\Theta_2}} \otimes id_{T_{s \cdot r} D_s \denot{\sigma}})
        ;
        st_{\denot{\Theta_1, \Delta_2, \Theta_2}, \denot{\sigma}}
      \end{multline*}
      by naturality of strength. By naturality, we can pull the first
      morphism out in front of $split$.
  \end{description}
\end{proof}
\fi

\begin{lemma}[Preservation] \label{lem:pres-sem}
  Let $\cdot \vdash e : \tau$ be a well-typed closed term, and suppose $e \mapsto e'$.
  Then there is a derivation of $\cdot \vdash e' : \tau$, and the semantics of both
  derivations are equal: $\denot{\vdash e : \tau} = \denot{\vdash e' : \tau}$.
\end{lemma}

\ifshort
\else
\begin{proof}
  By case analysis on the step rule, using the fact that $e$ is well-typed. For
  the beta-reduction steps for programs of non-monadic type, preservation
  follows by the soundness theorem; these cases are exactly the same as in Fuzz
  \citep{DBLP:conf/popl/AmorimGHKC17}. %

  The two step rules for programs of monadic type are new. It is possible to
  show soundness by appealing to properties of the graded monad $T_r$, but
  we can also show soundness more concretely by unfolding definitions and
  considering the underlying maps. 

  \begin{description}
    \item[Let-Bind $q$.] Suppose that $e = \letbind(\ret v,
      x. f)$ is a well-typed program with type $M_{s \cdot r + q}
      \tau$. Since subsumption is admissible~(\Cref{lem:subsump}), we may assume
      that the last rule is \textsc{Let-Bind} and we have derivations $\cdot \vdash
      \ret v : M_r \sigma$ and $x :_s \sigma \vdash f : M_{q}
      \tau$. By definition, the semantics of $\cdot \vdash e : M_{s \cdot r +
      q} \tau$ is given by the composition:
\[\begin{tikzcd}
	I & {D_s I} && {D_s T_r \sigma} & {T_{s \cdot r} D_s \sigma} & {T_{s \cdot r} T_q \tau} & {T_{s \cdot r + q} \tau}
	\arrow["{\lambda_{s, r, \sigma}}", from=1-4, to=1-5]
	\arrow[from=1-1, to=1-2]
	\arrow["{T_{s \cdot r} f}", from=1-5, to=1-6]
	\arrow["{\mu_{s \cdot r, q, \tau}}", from=1-6, to=1-7]
	\arrow["{D_s (v; \eta_\sigma; (0 \leq r)_\sigma)}", from=1-2, to=1-4]
\end{tikzcd}\]
      By substitution~(\Cref{lem:subst-sem}), we have a derivation of $\cdot
      \vdash f[v/x] : M_{q} \tau$. By applying the subsumption rule, we have a
      derivation of $\cdot \vdash f[v/x] : M_{s \cdot r + q} \tau$ with
      semantics:
\[\begin{tikzcd}
  I & {D_s I} & {D_s \sigma} & {T_q \tau} & {T_{s \cdot r + q} \tau}
  \arrow["{\lambda_{q, s \cdot r + q, \tau}}", from=1-4, to=1-5]
	\arrow["f", from=1-3, to=1-4]
	\arrow["{D_s v}", from=1-2, to=1-3]
	\arrow[from=1-1, to=1-2]
\end{tikzcd}\]
    Noting that the underlying maps of $s$ and $\mu$ are the identity function,
    both compositions have the same underlying maps, and hence are equal
    morphisms.
    \item[Let-Bind Assoc.] Suppose that $e =
      \letbind(\letbind(\rnd k, x. f), y. g)$ is a
      well-typed program with type $M_{s \cdot r + q} \tau$. Since
      subsumption is admissible~(\Cref{lem:subsump}), we have derivations:
      \[
        \vdash \rnd k : M_{r_1} \num
        \qquad\qquad
        x :_t \num \vdash f : M_{r_2} \sigma
        \qquad\qquad
        y :_s \sigma \vdash g : M_{q} \tau
      \]
      such that $t \cdot r_1 + r_2 = r$. By applying
      \textsc{Let-Bind} on the latter two derivations, we have:
      \[
        x :_{s \cdot t} \num \vdash \letbind(f, y. g) : M_{s \cdot r_2 + q} \tau
      \]
      And by applying \textsc{Let-Bind} again, we have:
      \[
        \vdash \letbind(\rnd k, x.  \letbind(f, y. g))
        : M_{s \cdot t \cdot r_1 + s \cdot r_2 + q} \tau
      \]
      This type is precisely $M_{s \cdot r + q} \tau$. The semantics of $e$ and
      $e'$ have the same underlying maps, and hence are equal morphisms.
      \qedhere
  \end{description}
\end{proof}
\fi

\subsection{Error Soundness}\label{subsec:sound}

The metric semantics interprets each program as a non-expansive map. We aim to
show that values of monadic type $M_r \sigma$ are interpreted as pairs of
values, where the first value is the result under an ideal operational semantics
and the second value is the result under an approximate, or finite-precision
(FP) operational semantics.

To make this connection precise, we first define the ideal and FP operational
semantics of our programs, refining our existing operational semantics so that
the rounding operation steps to a number.  Then, we define two denotational
semantics of our programs capturing the ideal and FP behaviors of programs, and
show that the ideal and FP operational semantics are sound with respect to this
denotation. Finally, we relate our metric semantics with our ideal and FP
semantics, showing how well-typed programs of monadic type satisfy the error
bound indicated by their type.

\paragraph{Ideal and FP operational semantics.}

We first refine our operational semantics to capture ideal and FP behaviors.

\begin{definition} \label{def:id-fp-steps}
  We define two step relations $e \mapsto_{id} e'$ and $e \mapsto_{fp} e'$ by
  augmenting the operational semantics with the following rules:
  \begin{align*}
    \rnd k \mapsto_{id} \ret k
    \qquad \text{and} \qquad
    \rnd k \mapsto_{fp} \ret \rho(k)
  \end{align*}
\end{definition}

Note that $\letbind(\rnd k, x.f)$ is no longer a value under these semantics,
since $\rnd k$ can step. Also note that these semantics are deterministic, and
by a standard logical relations argument, all well-typed terms normalize.

\paragraph{Ideal and FP denotational semantics.}

Much like our approach in $\Met$, we next define a denotational semantics of our
programs so that we can abstract away from the step relation.  
We develop both the ideal and approximate 
semantics in $\Set$, where maps are not required to be non-expansive.

\begin{definition} \label{def:id-fp-sem}
  Let $\Gamma \vdash e : \tau$ be a well-typed program. We can define two
  semantics in $\Set$:
  \begin{align*}
    \pdenot{\Gamma \vdash e : \tau}_{id} : \pdenot{\Gamma}_{id} \to \pdenot{\tau}_{id}
    \qquad\qquad
    \pdenot{\Gamma \vdash e : \tau}_{fp} : \pdenot{\Gamma}_{fp} \to \pdenot{\tau}_{fp}
  \end{align*}
  We take the graded comonad $D_s$ and the graded monad $T_r$ to both be the
  identity functor on $\Set$:
  \begin{align*}
    \pdenot{M_u~\tau}_{id} = \pdenot{\bang{s}~\tau}_{id} \triangleq \pdenot{\tau}_{id}
    \qquad\qquad
    \pdenot{M_u~\tau}_{fp} = \pdenot{\bang{s}~\tau}_{fp} \triangleq \pdenot{\tau}_{fp}
  \end{align*}

  The ideal and floating point interpretations of well-typed programs are
  straightforward, by induction on the derivation of the typing judgment. The
  only interesting case is for \textsc{Round}:
  \begin{align*}
    \pdenot{\Gamma \vdash \rnd k : M_{\rnderr} \num}_{id}
    \triangleq \pdenot{\Gamma \vdash k : \num}_{id}
    \qquad
    \pdenot{\Gamma \vdash \rnd k : M_{\rnderr} \num}_{fp}
    \triangleq \pdenot{\Gamma \vdash k : \num}_{fp} ; \rho
  \end{align*}
  where $\rho : R \to R$ is the rounding function.
\end{definition}

Following the same approach as in \Cref{lem:pres-sem}, it is straightforward to
prove that these denotational semantics are sound for their respective
operational semantics.

\begin{lemma}[Preservation] \label{lem:pres-id-fp}
  Let $\cdot \vdash e : \tau$ be a well-typed closed term, and suppose $e
  \mapsto_{id} e'$. Then there is a derivation of $\cdot \vdash e' : \tau$ and
  the semantics of both derivations are equal: $\pdenot{\vdash e : \tau}_{id} =
  \pdenot{\vdash e' : \tau}_{id}$. The same holds for the FP denotational and
  operational semantics.
\end{lemma}
\ifshort
\else
\begin{proof}
  By case analysis on the step relation. We detail the cases where $e = \rnd k$.
  \begin{description}
    \item[Case: $\rnd k \mapsto_{id} \ret k$.] Suppose that $\cdot \vdash \rnd k
      : M_u \num$, where $\rnderr \leq u$. Then there is a derivation of $\cdot
      \vdash \ret k : M_u \num$ by \textsc{Return} and \textsc{Subsumption}, and
      \[
        \pdenot{\cdot \vdash \rnd k : M_u \num}_{id}
        = \pdenot{\cdot \vdash k : \num}_{id}
        = \pdenot{\cdot \vdash \ret k : M_u \num}_{id} .
      \]
    \item[Case: $\rnd k \mapsto_{fp} \ret \rho(k)$.] Suppose that $\cdot \vdash
      \rnd k : M_u \num$, where $\rnderr \leq u$. Then there is a derivation of
      $\cdot \vdash \ret k : M_u \num$ by \textsc{Return} and
      \textsc{Subsumption}, and
      \[
        \pdenot{\cdot \vdash \rnd k : M_u \num}_{fp}
        = \pdenot{\cdot \vdash \rho(k) : \num}_{fp}
        = \pdenot{\cdot \vdash \ret \rho(k) : M_u \num}_{fp} .
      \]
      \qedhere
  \end{description}
\end{proof}
\fi

\paragraph{Establishing error soundness.}

Finally, we connect the metric semantics with the ideal and FP semantics.  Let
$U : \Met \to \Set$ be the forgetful functor mapping each metric space to its
underlying set, and each morphism of metric spaces to its underlying function on
sets. We have:

\begin{lemma}[Pairing] \label{lem:pairing}
  Let $\cdot \vdash e : M_r \num$. Then we have:
  $
    U \denot{e} = \langle \pdenot{e}_{id}, \pdenot{e}_{fp} \rangle
  $
  in $\Set$: the first projection of $U \denot{e}$ is $\pdenot{e}_{id}$, and the
  second projection is $\pdenot{e}_{fp}$.
\end{lemma}
\ifshort
\else
\begin{proof}
  By the logical relation for termination, the judgment $\cdot \vdash e : M_r
  \num$ implies that $e$ is in $\mathcal{R}_{M_r \num}^n$ for some $n \in
  \mathbb{N}$. We proceed by induction on $n$.

  For the base case $n = 0$, we know that $e$ reduces to either $\ret v$
  or $\rnd v$. We can conclude since by inversion $v$ must be a real
  constant, and $U\denot{v} = \pdenot{v}_{id} = \pdenot{v}_{fp}$.

  For the inductive case $n = m + 1$, we know that $e$ reduces to
  $\letbind(\rnd k, x. f)$. By the logical relation, we know
  that for all $\cdot \vdash v : \num$, we have $f[v/x] \in R_{M_r \num}^m$ and so
  by induction:
  \begin{align*}
    \pdenot{f[k/x]}_{id}
    &\triangleq \pdenot{\letbind(\rnd k, x. f)}_{id}
    = U \denot{f[k/x]} ; \pi_1 \\
    \pdenot{f[\rho(k)/x]}_{fp}
    &\triangleq \pdenot{\letbind(\rnd k, x. f)}_{fp}
    = U \denot{f[\rho(k)/x]} ; \pi_2
  \end{align*}
  Thus we just need to show:
  \[
    U\denot{\letbind(\rnd k, x.f)}
    = \langle U\denot{f[k/x]}; \pi_1, U\denot{f[\rho(k)/x]}; \pi_2 \rangle
  \]
  where we have judgments $\cdot \vdash \rnd k : M_r \num$ and $x :_s
  \num \vdash f : M_q \num$. Unfolding the definition,
  $\denot{\letbind(\rnd k, x. f)}$ is equal to the composition:
  \[\begin{tikzcd}
    I & {D_s I} & {D_s T_r \denot{\num}} & {T_{r \cdot s} D_s \denot{\num}} & {T_{r \cdot s} T_{q} \denot{\num}} & {T_{r \cdot s + q} \denot{\num}}
    \arrow["{m_{s, I}}", from=1-1, to=1-2]
    \arrow["{D_s \denot{\rnd k}}", from=1-2, to=1-3]
    \arrow["{\lambda_{s, r, \denot{\num}}}", from=1-3, to=1-4]
    \arrow["{T_{r \cdot s} \denot{f}}", from=1-4, to=1-5]
    \arrow["{\mu_{r \cdot s, q, \denot{\num}}}", from=1-5, to=1-6]
  \end{tikzcd}\]
  We conclude by applying the substitution lemma and considering the underlying
  maps.
  \qedhere
\end{proof}
\fi

As a corollary, we have soundness of the error bound for programs with monadic type.

\begin{corollary}[Error soundness]\label{cor:err-sound}
  Let $\cdot \vdash e : M_r \num$ be a well-typed program. Then $e
  \mapsto_{id}^* \ret v_{id}$ and $e \mapsto_{fp}^* \ret v_{fp}$ such that
  $d_{\denot{\num}}(\denot{v_{id}}, \denot{v_{fp}}) \leq r$.
\end{corollary}
\ifshort
\else
\begin{proof}
  Under the ideal and floating point semantics, the only values of monadic type
  are of the form $\ret v$. Since these operational semantics are
  type-preserving and normalizing, we must have $e \mapsto_{id}^* \ret v_{id}$ and
  $e \mapsto_{fp}^* \ret v_{fp}$. By soundness~(\Cref{lem:pres-id-fp}), 
   $\pdenot{e}_{id} = \pdenot{\ret
  v_{id}}_{id}$ and $\pdenot{e}_{fp} = \pdenot{\ret v_{fp}}_{fp}$. By pairing
  (\Cref{lem:pairing}), we have $U\denot{e} = \langle \pdenot{\ret v}_{id},
  \pdenot{\ret v}_{fp} \rangle$. Since the forgetful functor is the identity on
  morphisms, we have $\denot{e} = \langle \pdenot{\ret v_{id}}_{id},
  \pdenot{\ret v_{fp}}_{fp} \rangle : I \to T_r \denot{\num}$ in $\Met$. Now by
  definition $\pdenot{\ret v_{id}}_{id} = \pdenot{v_{id}}_{id}$ and
  $\pdenot{\ret v_{fp}}_{fp} = \pdenot{v_{fp}}_{fp}$, hence we can conclude by
  the definition of the monad $T_r$.
\end{proof}
\fi

\section{Case Studies}\label{sec:examples}
To illustrate how \Lang~ can be used to bound 
the sensitivity and roundoff error of numerical programs,  
we must fix our interpretation of the numeric type
$\denot{\num}$ using an appropriate metric space $(R,d_R)$ and augment our
language to include primitive arithmetic operations over the set $R$.

\begin{figure}
\centering
\begin{minipage}[t]{.4\textwidth}
  \centering
\begin{lstlisting}[
        xleftmargin=.1\textwidth,
        xrightmargin=.1\textwidth]
add : (num $\times$ num) -o num   
mul : (num $\otimes$ num) -o num
div : (num $\otimes$ num) -o num    
sqrt : ![0.5]num -o num 
\end{lstlisting}
    \caption{Primitive operations in \Lang, typed using the relative precision (RP) metric.}
    \label{fig:prim_ops}
\end{minipage}%
\hfill
\begin{minipage}[t]{.5\textwidth}
  \centering
\begin{lstlisting}[
        xleftmargin=.1\textwidth,
        xrightmargin=.1\textwidth]
addfp : (num $\times$ num) -o M[eps]num   
mulfp : (num $\otimes$ num) -o M[eps]num
divfp : (num $\otimes$ num) -o M[eps]num    
sqrtfp : ![0.5]num -o M[eps]num 
\end{lstlisting}
    \caption{Type signatures of defined operations that 
	perform rounding in \Lang; \lstinline{eps} denotes the unit roundoff.}
    \label{fig:rounded_ops}
\end{minipage}
\end{figure}
\begin{figure}
     \centering
\begin{subfigure}[t]{0.5\textwidth}     
\begin{lstlisting}[
        xleftmargin=.1\textwidth,
        xrightmargin=.1\textwidth]
function mulfp (xy: (num, num)) 
  : M[eps]num { 
  s = mul xy;
  rnd s
}
\end{lstlisting}
\end{subfigure}
     \hfill
\begin{subfigure}[t]{0.5\textwidth}     
\begin{lstlisting}[
        xleftmargin=.1\textwidth,
        xrightmargin=.1\textwidth]
function addfp (xy: <num, num>) 
  : M[eps]num { 
  s = add xy;
  rnd s
}
\end{lstlisting}
\end{subfigure}
    \caption{Example defined operations that perform rounding in \Lang. 
		We denote the unit roundoff by \lstinline{eps}.}
    \label{fig:def_rounded_ops}
\end{figure}

If we interpret our numeric type $\num$ as the set
of strictly positive real numbers $\R^{>0}$ with the relative precision (RP)
metric (\Cref{def:rp}), then we can use 
\Lang to perform a relative error analysis as described by \citet{Olver}.
Using this metric, we extend the language with the four primitive arithmetic operations
shown in \Cref{fig:prim_ops}.

Recall that our metric semantics interprets each
program as a non-expansive map. If we take the semantics of the arithmetic
operations as being the standard addition and multiplication of positive real
numbers, then \lstinline{add} and \lstinline{mul} as defined in 
\Cref{fig:prim_ops} are non-expansive functions
\citep[Corollary 1 \& Property V]{Olver}; recall that the two product types have
different metrics (\Cref{subsec:interp}).

Using the primitive operations in \Cref{fig:prim_ops} and the $\mathbf{rnd}$ construct, 
we can write functions for the basic arithmetic operations in ~\Lang that 
perform concrete rounding when interpreted according to the FP semantics. The type 
signature of these functions is shown in \Cref{fig:rounded_ops},
 and implementations 
of a multiplication and addition that perform rounding are shown in \Cref{fig:def_rounded_ops}.  

The examples presented in this section use the
actual syntax of an implementation of \Lang, which is introduced in 
\Cref{sec:implementation}. The implementation closely follows the syntax of the language
as presented in \Cref{fig:typing_rules}, with some additional syntactic sugar.  
For instance, we write 
\lstinline{(x = v; e)} to denote $\tlet x = v ~ \tin e$, and \lstinline{(let x = v; f)} to denote 
$\letbind (v,x.f)$. For top level programs, we write \lstinline{(function ID args} $\{v\}$ \lstinline{e)} to denote 
$\tlet \text{ID} = v ~ \tin e$, where $v$ is a lambda term with arguments \lstinline{args}. We write pairs 
of type $ - \times - $ and $ - \tensor - $ as \lstinline{(|-,-|)} and \lstinline{(-,-)}, respectively.
Finally, for types, we write \lstinline{M[u]num} to represent monadic
types with a numeric grade \lstinline{u} and we write \lstinline{![s]} to represent exponential types
with a numeric grade \lstinline{s}.

\paragraph{Choosing the RP Rounding Function}
Recall that we require the rounding function $\rho$ to
be a function such that for every $x \in R^{>0}$, we have $RP(x, \rho(x)) \le
\epsilon$; that is, the rounding function must satisfy an accuracy guarantee
with respect to the metric $RP$ on the set $ R^{>0}$. If we choose $\rho_{RU}:
R^{>0} \rightarrow R^{>0}$ to be rounding towards $+\infty$, then by \cref{eq:RPbnd2} we 
have that $RP(x, \rho(x)) \le$ \lstinline{eps} where \lstinline{eps} is the unit 
roundoff. Error soundness (\Cref{cor:err-sound}) implies
that for the functions \lstinline{mulfp} 
and \lstinline{addfp}, the results of the ideal and approximate computations
differ by at most \lstinline{eps}. 

\paragraph*{Underflow and overflow.}
In the following examples, \emph{we assume that
the results of computations do not overflow or underflow}. Recall from \Cref{sec:background}
that the standard model for floating-point arithmetic given in \cref{eq:op_model} is only 
valid under this assumption.
In \Cref{sec:extensions},
we discuss how \Lang can be extended to handle overflow, underflow, and exceptional values. 

\paragraph{Example: The Fused Multiply-Add Operation}
We warm up with a simple example of a \emph{multiply-add} (MA) operation: given $x,
y, z$, we want to compute $x * y + z$.  The \Lang implementation of \lstinline{MA} is 
given in \Cref{fig:FMA}. The index \lstinline{2*eps} on the return type indicates 
that the roundoff error is at most twice
the unit roundoff,  due to the two separate rounding
operations in \lstinline{mulfp} and \lstinline{addfp}.

Multiply-add is extremely common in numerical code, and modern
architectures typically support a \emph{fused} multiply-add (FMA) operation.
This operation performs a multiplication followed by an addition, $x*y+z$, as
though it were a single floating-point operation. The FMA operation therefore
incurs a single rounding error, rather than two.  The \Lang implementation of 
a FMA operation is given in \Cref{fig:FMA}. The index on the return type of the function 
is \lstinline{eps}, reflecting a reduction in the roundoff error when compared to 
the function \lstinline{MA}.

\begin{figure}
\centering
\begin{minipage}[t]{.5\textwidth}
  \centering
\begin{lstlisting}
function MA (x: num) (y: num) (z: num) 
  : M[2*eps]num { 
  s = mulfp (x,y);
  let a = s;
  addfp (|a,z|)
}
\end{lstlisting}
\end{minipage}\hfill
\begin{minipage}[t]{.5\textwidth}
  \centering       
\begin{lstlisting}
function FMA (x: num) (y: num) (z: num) 
  : M[eps]num {
  a = mul (x,y);
  b = add (|a,z|);
  rnd b
}
\end{lstlisting}
\end{minipage}
    \caption{Multiply-add and fused multiply-add in \Lang.}
    \label{fig:FMA}
\end{figure}

\paragraph{Example: Evaluating Polynomials}
A standard method for evaluating a polynomial
is Horner's scheme, which rewrites an $n$th-degree polynomial
$p(x) = a_0+a_1x+\cdots a_n x^n$ as 
\[ p(x) = a_0 + x(a_1 + x(a_2  + \cdots x(a_{n-1} + ax_n)\cdots)),\]
and computes the result using only $n$ multiplications and $n$ additions.
Using \Lang, we can perform an error analysis on a version of Horner's scheme
that uses a FMA operation to evaluate second-order polynomials of the form
$p(\vec{a},x) = a_2x^2 + a_1x + a_0$ where $x$ and all $a_i$s are non-zero
positive constants.  The implementation \lstinline{Horner2} in \Lang is given in \Cref{fig:Horner} and 
shows that the rounding error on exact inputs is guaranteed to be bounded by \lstinline{2*eps}: 
\begin{equation} 
RP(\pdenot{\mathbf{Horner2}~as~x}_{id},
\pdenot{\mathbf{Horner2}~as~x}_{fp}) \le 2* \texttt{eps}. \label{eq:horner_er}
\end{equation} 
\begin{figure}
\centering
\begin{minipage}[t]{.5\textwidth}
\begin{lstlisting}
function Horner2 
  (a0: num) (a1: num) 
  (a2: num) (x: ![2.0]num) 
  : M[2*eps]num {
  let [x1] = x ;
  s1 = FMA a2 x' a1;
  let z = s1;
  FMA z x1 a0
}
\end{lstlisting}
\end{minipage}\hfill
\begin{minipage}[t]{.5\textwidth}       
\begin{lstlisting}
function Horner2_with_error 
  (a0: M[eps]num) (a1: M[eps]num) 
  (a2: M[eps]num) (x: ![2.0]M[eps]num) 
  : (M[7*eps]num) { 
  let [x1] = x ;
  let a0' = a0; let a1' = a1;
  let a2' = a2; let x' = x1; 
  s1 = FMA a2' x' a1';
  let z = s1;
  FMA z x' a0'
}
\end{lstlisting}
\end{minipage}
    \caption{Horner's scheme for evaluating a second order polynomial in \Lang with 
    	(\lstinline{Horner2_with_error}) and without (\lstinline{Horner2}) input error.}
    \label{fig:Horner}
\end{figure}

\paragraph{Example: Error Propagation and Horner's Scheme}
As a consequence of the metric interpretation of programs
(\Cref{subsec:interp}), the type of \lstinline{Horner2}
also guarantees bounded sensitivity of the ideal semantics, which corresponds to
$p(\vec{a},x) = a_2x^2 + a_1x + a_0$. Thus for any ${a}_i,{a}'_i,x, x' \in
R^{>0}$, we can measure the sensitivity of \lstinline{Horner2} to rounding errors
introduced by the inputs: if $x'$ is an approximation to $x$ of RP $q$, and each
$a_i$ is an approximation to its corresponding $a_i$ of RP $r$, then 
\begin{equation} 
RP(p(\vec{a},x), p(\vec{a}',x')) \le \sum_{i=0}^2 RP(a_i, a_i')
+ 2 \cdot RP(x, x') \le 3r + 2q. \label{eq:horner_sens} 
\end{equation} 
The term $2q$ reflects that \lstinline{Horner2} is $2$-sensitive 
in the variable $x$. The fact that we take the sum of the approximation
distances over the $a_i$'s  follows from the metric on
the function type (\Cref{subsec:interp}). 

The interaction between the sensitivity of the function
under its ideal semantics and the rounding error incurred by 
\lstinline{Horner2} over exact inputs is made clear by the function 
\lstinline{Horner2_with_error}, shown in \Cref{fig:Horner}. 
From the type, we see that the total roundoff error of 
\lstinline{Horner2_with_error} is \lstinline{7*eps}: from
\cref{eq:horner_sens} it follows that the sensitivity of the function 
contributes \lstinline{5*eps}, and rounding error incurred by 
evaluating \lstinline{Horner2} over exact inputs contributes the 
remaining \lstinline{2*eps}.


\ifshort
\else
\subsection{Floating-Point Conditionals} 
In the presence of rounding error, conditional branches present a particular challenge:
while the ideal execution may follow one branch, the floating-point execution may follow another.
In \Lang, we can perform rounding error analysis on programs with conditional expressions  (case analysis)
when executions \emph{take the same branch}, for instance, when the data in the conditional
is a boolean expression that does not  have floating-point error because it is some kind of 
parameter to the system, or some exactly-represented value that is computed only from other
exactly-represented values.  This is a restriction of the Fuzz-style type system of \Lang, which is not able to compare the 
difference between two different branches since the main metatheoretic guarantee 
only serves as a sensitivity analysis describing how a single program behaves on two different inputs. 
In \Lang, the rounding error of a program with a case analysis is then a
measure of the maximum rounding error that occurs in any single branch.

As an example of performing rounding error analysis in \Lang~ on functions with conditionals, we first add the 
primitive operation $\mathbf{is\_pos}: !_\infty \R \multimap \mathbb{B}$, which tests if a real number is greater 
than zero. The sensitivity on the argument to $\mathbf{is\_pos}$ is necessarily  
infinity, since an arbitrarily small  change in the argument to could lead to an infinitely large change in 
the boolean output. Using $\mathbf{is\_pos}$
we define the function $\mathbf{case1}$, which computes the square
of a negative number, or returns the value 0 (lifted to monadic type):
\[ \mathbf{case1} : ~ !_\infty \R \multimap M_u\R\]
\begin{align*} 
\mathbf{case1}~ x = ~ &\tlet ~[c] = ~\mathbf{is\_pos}~ x ~\tin \\
& \mathbf{if } ~c ~\mathbf{then}~ \mathbf{mul}_{fp} (x,x) ~\mathbf{ else } ~ \ret 0.
\end{align*} 
From the signature of $\mathbf{case1}$, we see that the relative distance (RD) is unit roundoff, due to the single 
rounding in $\mathbf{mul}_{fp} (x,x) $. 
\fi

\section{Implementation and Evaluation}\label{sec:implementation}

\subsection{Prototype Implementation}
We have developed a prototype type-checker for \Lang in OCaml, based on the
sensitivity-inference algorithm due to \citet{10.1145/2746325.2746335} developed
for {DFuzz}~\cite{dfuzz}, a dependently-typed extension of Fuzz.
Given an environment $\Gamma$, a term $e$, and a type $\sigma$, the goal of
type checking is to determine if a derivation $\Gamma \vdash e : \sigma$ exists.
For sensitivity type systems, type checking and type inference can 
be achieved by solving the sensitivity inference problem. The
sensitivity inference problem is defined using \emph{context skeletons}
$\Gamma^{\bullet}$ which are partial maps from variables to \Lang types. If we
denote by $\overline{\Gamma}$ the context $\Gamma$ with all sensitivity
assignments removed, then the sensitivity inference problem is
defined~\cite[Definition 5]{10.1145/2746325.2746335} as follows.

\begin{definition}[Sensitivity Inference]
Given a skeleton $\Gamma^{\bullet}$ and a term $e$, the 
\emph{sensitivity inference problem} computes an environment $\Gamma$ and
a type $\sigma$ with a derivation $\Gamma \vdash e : \sigma$ such that 
$\Gamma^{\bullet} = \overline{\Gamma}$. \end{definition}
\ifshort \else
\begin{figure}[htbp]
\begin{center}
\AXC{}
\RightLabel{(Var)}
\UIC{$\Gamma^{\bullet}, x: \sigma; x \Rightarrow \Gamma^0, x :_1\sigma;\sigma$}
\bottomAlignProof
\DisplayProof
\hskip 1em
\AXC{$\Gamma^{\bullet}; v  \Rightarrow \Delta; \sigma \tensor \tau$ }
\noLine \UIC{$\Gamma^{\bullet},x: \sigma,y:\tau; e \Rightarrow \Gamma, x:_{s_1}\sigma,  y:_{s_2} \tau; \rho $}
\RightLabel{($\tensor$ E)}
\UIC{$ \Gamma^{\bullet}; \tlet (x,y) \ = \ v \ \tin e  \Rightarrow \Gamma + \max(s_1,s_2)* \Delta; \rho$}
\bottomAlignProof
\DisplayProof
\vskip 1em

\AXC{$\Gamma^{\bullet};  v \Rightarrow \Gamma; \tau_1 \times \tau_2$}
\RightLabel{($\times$ E)}
\UIC{$\Gamma^{\bullet}; {\pi}_i \ v \Rightarrow \Gamma; \tau_i$}
\bottomAlignProof
\DisplayProof
\hskip 0.5em
\AXC{$\Gamma^{\bullet}, x : \sigma; e \Rightarrow \Gamma, x:_{s}\sigma; \tau$}
\noLine \UIC{$s \ge 1$}
\RightLabel{($\multimap$ I)}
\UnaryInfC{$\Gamma^{\bullet}; \lambda (x : \textcolor{blue}{\sigma}).e \Rightarrow \Gamma; \sigma \multimap \tau $}
\bottomAlignProof
\DisplayProof
\hskip 0.5em
\AXC{$\Gamma^{\bullet}; v \Rightarrow \Gamma; \sigma \multimap \tau$}
\noLine \UIC{$\Gamma^{\bullet}; w \Rightarrow \Delta;\sigma' $}
\noLine \UIC{ $\sigma' \sqsubseteq \sigma$}
\RightLabel{($\multimap$ E)}
\UnaryInfC{$\Gamma^{\bullet}; v w \Rightarrow \Gamma + \Delta; \tau $}
\bottomAlignProof
\DisplayProof
\vskip 1em

\AXC{$\Gamma^{\bullet}; e \Rightarrow \Gamma;  \tau$}
\noLine \UIC{$\Gamma^{\bullet},x ; f \Rightarrow \Theta, x:_{s} \tau; \sigma$}
\AXC{}
\noLine \UIC{$s > 0$}
\RightLabel{(Let)}
\BIC{$\Gamma^{\bullet}; \tlet x  = e \ \tin f \Rightarrow s * \Gamma + \Theta; \sigma$}
\bottomAlignProof
\DisplayProof
\hskip 1em
\AXC{$\Gamma^{\bullet}; v \Rightarrow \Gamma_1;  \sigma$}
\AXC{}
\noLine \UIC{$\Gamma^{\bullet}; w \Rightarrow \Gamma_2; \tau$}
\RightLabel{($\times$ I)}
\BIC{$\Gamma^{\bullet}; \langle v, w \rangle \Rightarrow \max({\Gamma_1, \Gamma_2}); \sigma \times \tau $}
\bottomAlignProof
\DisplayProof
\vskip 1em


\AXC{$\Gamma^{\bullet}; v \Rightarrow \Gamma; \sigma$ }
\RightLabel{($+$ I)}
\UIC{$\Gamma^{\bullet}; \inl \ v \Rightarrow \Gamma; \sigma + \tau$}
\bottomAlignProof
\DisplayProof
\hskip 0em
\AXC{$\Gamma^{\bullet}; v \Rightarrow \Gamma; {!_s \sigma}$}
\noLine \UIC{$\Gamma^{\bullet}; x \Rightarrow \Theta, x:_{t*s} \sigma ; e : \tau$}
\RightLabel{($!$ E)}
\UIC{$\Gamma^{\bullet} ; \tlet [x] = v \ \tin e \Rightarrow t*\Gamma + \Theta; \tau$}
\bottomAlignProof
\DisplayProof
\hskip 0em
\AXC{$\Gamma^{\bullet} ; v \Rightarrow \Gamma; \sigma$}
\noLine \UIC{$\{ \mathbf{op} :\sigma \lin \num \} \in {\Sigma}$}
\RightLabel{(Op)}
\UIC{$\Gamma^{\bullet}; \mathbf{op}(v) \Rightarrow \Gamma; \num$}
\bottomAlignProof
\DisplayProof
\vskip 1em



\AXC{$\Gamma^{\bullet}, x; e \Rightarrow \Theta,  x:_s \sigma; \rho_1$}
\noLine \UIC{$\Gamma^{\bullet},y ; f \Rightarrow \Theta,  y:_s \tau; \rho_2$}
\AXC{}
\noLine \UIC{$\Gamma^{\bullet}; v \Rightarrow \Gamma;  \sigma+\tau$}
\RightLabel{($+$ E)}
\BIC{$\Gamma^{\bullet}; \mathbf{case} \ v \ \mathbf{of} \ (\inl x.e \ | \ \inr x.f) 
\Rightarrow \bar{s} * \Gamma + \Theta; \max(\rho_1,\rho_2)$}
\bottomAlignProof
\DisplayProof
\hskip 1em
\AXC{$\Gamma^{\bullet}; v \Rightarrow \Gamma; \sigma$ }
\RightLabel{($!$ I)}
\UIC{$\Gamma^{\bullet}; [v\textcolor{blue}{\{s\}}] \Rightarrow s * \Gamma; {!_s \sigma}$}
\bottomAlignProof
\DisplayProof
\vskip 1em

\AXC{}
\noLine \UIC{$\Gamma^{\bullet};v  \Rightarrow \Gamma_1; \sigma $}
\noLine \UIC{$\Gamma^{\bullet}; w \Rightarrow \Gamma_2; \tau$}
\RightLabel{($\tensor$ I)}
\UIC{$\Gamma;(v,w) \Rightarrow \Gamma_1 + \Gamma_2; \sigma \tensor \tau$}
\bottomAlignProof
\DisplayProof
\hskip 0.5em
\AXC{$\Gamma^{\bullet}; v \Rightarrow \Gamma; \tau$}
\RightLabel{(Ret)}
\UIC{$\Gamma^{\bullet}; \ret v \Rightarrow \Gamma; M_0 \tau$}
\bottomAlignProof
\DisplayProof
\hskip 0.5em
\AXC{$\Gamma^{\bullet}; v \Rightarrow \Gamma; \num$}
\RightLabel{(Rnd)}
\UIC{$\Gamma^{\bullet}; \rnd v \Rightarrow \Gamma; M_q \num$}
\bottomAlignProof
\DisplayProof
\vskip 1em

\AXC{$\Gamma^{\bullet}; v  \Rightarrow \Gamma; M_r \sigma$}
\noLine \UIC{$\Gamma^{\bullet}, x; f \Rightarrow \Theta, x:_{s} \sigma;  M_{q} \tau$}
\RightLabel{($M_u$ E)}
\UIC{$\Gamma^{\bullet}; \letbind(v,x.f) \Rightarrow s * \Gamma + \Theta;  M_{s*r+q} \tau$}
\bottomAlignProof
\DisplayProof
\hskip 0.5em
\AXC{}
\RightLabel{(Const)}
\UIC{$\Gamma^{\bullet};k \in R \Rightarrow \Gamma^0; \num$}
\bottomAlignProof
\DisplayProof
\vskip 1em

\begin{equation*}
\bar{s} \triangleq
\begin{cases}
    s &  s > 0 \\ \epsilon &  \text{otherwise} 
\end{cases}
\end{equation*}

\end{center}
    \caption{Algorithmic rules for \Lang.  $\Gamma^0$ denotes  an environment where all variables have sensitivity zero. The supertype ($\max$) and subtype ($\sqsubseteq$) relations on \Lang types are given in \Cref{fig:max_ty} and \Cref{fig:sub_ty}, respectively.}
    \label{fig:alg_rules}
\end{figure}

\begin{figure}[htbp]
\footnotesize

\begin{equation*}
\begin{aligned}[c]
\max(\sigma_1 \times \tau_1, \sigma_2 \times \tau_2) &\triangleq
\max(\sigma_1, \sigma_2) \times \max(\tau_1,\tau_2) \\ 
\max(\sigma_1 \otimes \tau_1, \sigma_2 \otimes \tau_2) &\triangleq
\max(\sigma_1, \sigma_2) \otimes \max(\tau_1,\tau_2) \\ 
\max(\sigma_1 + \tau_1, \sigma_2 + \tau_2) &\triangleq
\max(\sigma_1, \sigma_2) + \max(\tau_1,\tau_2) \\
\max(\text{M}_{s_1}\tau_1,\text{M}_{s_2}\tau_2) &\triangleq
\text{M}_{\max(s_1,s_2)}\max(\tau_1,\tau_2) \\
\max(!_{s_1}\tau_1,!_{s_2}\tau_2) &\triangleq
\ !_{\min(s_1,s_2)}\max(\tau_1,\tau_2)
\end{aligned}
\qquad
\begin{aligned}[c]
\min(\sigma_1 \times \tau_1, \sigma_2 \times \tau_2) &\triangleq
\min(\sigma_1, \sigma_2) \times \min(\tau_1,\tau_2) \\ 
\min(\sigma_1 \otimes \tau_1, \sigma_2 \otimes \tau_2) &\triangleq
\min(\sigma_1, \sigma_2) \otimes \min(\tau_1,\tau_2) \\ 
\min(\sigma_1 + \tau_1, \sigma_2 + \tau_2) &\triangleq
\min(\sigma_1, \sigma_2) + \min(\tau_1,\tau_2) \\
\min(\text{M}_{s_1}\tau_1,\text{M}_{s_2}\tau_2) &\triangleq
\text{M}_{\min(s_1,s_2)}\min(\tau_1,\tau_2) \\
\min(!_{s_1}\tau_1,!_{s_2}\tau_2) &\triangleq
\ !_{\max(s_1,s_2)}\min(\tau_1,\tau_2) 
\end{aligned}
\end{equation*}

\begin{align*}
\max(\sigma_1 \multimap \tau_1, \sigma_2 \multimap \tau_2) &\triangleq
\min(\sigma_1, \sigma_2) \multimap \max(\tau_1,\tau_2) \\
\min(\sigma_1 \multimap \tau_1, \sigma_2 \multimap \tau_2) &\triangleq
\max(\sigma_1, \sigma_2) \multimap \min(\tau_1,\tau_2) 
\end{align*}

    \caption{The $\max$ (supertype) relation on \Lang types, with $s_1,s_2 \in \NNR \cup \{\infty\}$. }
    \label{fig:max_ty}
\end{figure}

\begin{figure}[htbp]
\begin{center}
\small
\AXC{}
\RightLabel{$\sqsubseteq .$unit}
\UIC{\textbf{unit}$ \ \sqsubseteq$ \textbf{unit}}
\DisplayProof
\hskip 2.5em
\AXC{$$}
\RightLabel{$\sqsubseteq .$num}
\UIC{\textbf{num}$ \ \sqsubseteq$ \textbf{num}}
\DisplayProof
\vskip 1em

\AXC{$\sigma \sqsubseteq \sigma'$}
\AXC{$\tau \sqsubseteq \tau'$}
\RightLabel{$\sqsubseteq .\times$}
\BIC{$ \sigma \times \tau \sqsubseteq \sigma' \times \tau' $}
\DisplayProof
\hskip 2.5em
\AXC{$\sigma \sqsubseteq \sigma'$}
\AXC{$\tau \sqsubseteq \tau'$}
\RightLabel{$\sqsubseteq .\tensor$}
\BIC{$ \sigma \tensor \tau \sqsubseteq \sigma' \tensor \tau' $}
\DisplayProof
\hskip 2.5em
\AXC{$\sigma \sqsubseteq \sigma'$}
\AXC{$\tau \sqsubseteq \tau'$}
\RightLabel{$\sqsubseteq .+$}
\BIC{$ \sigma + \tau \sqsubseteq \sigma' + \tau' $}
\DisplayProof
\vskip 1em

\AXC{$\sigma' \sqsubseteq \sigma$}
\AXC{$\tau \sqsubseteq \tau' $}
\RightLabel{$\sqsubseteq .\multimap$}
\BIC{$ \sigma \multimap \tau \sqsubseteq \sigma' \multimap \tau' $}
\DisplayProof
\hskip 2.5em
\AXC{$\sigma \sqsubseteq \sigma'$}
\AXC{$u \le u'$}
\RightLabel{$\sqsubseteq .$M}
\BIC{$ \text{M}_u \sigma \sqsubseteq \text{M}_{u'} \sigma' $}
\DisplayProof
\hskip 2.5em
\AXC{$\sigma \sqsubseteq \sigma'$}
\AXC{$s \le s'$}
\RightLabel{$\sqsubseteq .!$}
\BIC{$ \bang{s'} \sigma \sqsubseteq !_{s} \sigma' $}
\DisplayProof

\end{center}
    \caption{The subtype relation in \Lang, with $s,s',u,u' \in \NNR \cup \{\infty\}$. }
    \label{fig:sub_ty}
\end{figure}

We solve the sensitivity inference problem using the algorithm given in 
 \Cref{fig:alg_rules}. 
\fi
Given a term $e$ and a skeleton 
environment $\Gamma^{\bullet}$,
the algorithm produces an environment $\Gamma^{\bullet}$ with
sensitivity information and a type $\sigma$. Calls to the algorithm are 
written as $\Gamma^{\bullet}; e \Rightarrow \Delta; \sigma$.  

\ifshort \else
Following~\cite{10.1145/2746325.2746335}
the algorithm uses a \emph{bottom-up} rather than a \emph{top-down} approach.  
In the \emph{top-down}
approach, given a term $e$, type $\sigma$, and environment $\Gamma$, the
environment is
split and used recursively to type the subterms of the expression $e$. 
The \emph{bottom-up} approach avoids splitting the environment $\Gamma$ by
calculating the minimal sensitivities and roundoff errors required to type each
subexpression.
The sensitivities and errors of each subexpression are then combined
and compared to $\Gamma$ and $\sigma$ using subtyping. The  subtyping relation in \Lang is defined in \Cref{fig:sub_ty} and captures the fact that a {$k$-sensitive} function is also {$k'$-sensitive} for $k \le k'$. Importantly, subtyping is admissible in \Lang.
\begin{theorem}\label{thm:sub_type}
The typing judgment $\Gamma \vdash e : \tau'$ is derivable
given a derivation $\Gamma \vdash e : \tau$ and a type $\tau'$
such that $\tau \sqsubseteq \tau'$. 
\end{theorem}
\begin{proof}
The proof follows by induction on the derivation $\Gamma \vdash e : \tau$. Most cases are immediate, but some require weakening (\Cref{thm:weakening}). We detail two examples here. 
\begin{description}
\item[Case (Var).] We are required to show $\Gamma,x:_s \tau,\Delta \vdash x : \tau'$ given $\tau \sqsubseteq \tau'$. We can conclude by weakening (\Cref{thm:weakening}) and the fact that the subenvironment relation is preserved by subtyping; i.e., $x:_s \tau' \sqsubseteq x:_s \tau$.  
\item[Case ($!$ I).] We are required to show 
$s * \Gamma \vdash [v] : \bang{s'}\sigma'$ for some $s'\le s$ and $\sigma' \sqsubseteq \sigma$. By the induction hypothesis we have $\Gamma \vdash v : \sigma'$, and by the box introduction rule ($!$ I) we have $s' * \Gamma \vdash [v] : \bang{s'} \sigma'$. Because $s * \Gamma \sqsubseteq s' * \Gamma$ for $s' \le s$ we can conclude by weakening.  
\end{description}
\end{proof}
\fi

The algorithmic rules presented in \Cref{fig:alg_rules} define a \emph{sound} type checking algorithm for \Lang:
\begin{theorem}[Algorithmic Soundness] 
If $ \ \Gamma^{\bullet};e \Rightarrow \Gamma; \sigma$ then 
there exists a derivation $\Gamma \vdash e : \sigma$. 
\end{theorem}
\ifshort
\else
\begin{proof}
By induction on the algorithmic derivations, we see that every step of the
algorithm corresponds to a derivation in \Lang. Many cases are immediate, but
those that use subtyping, supertyping, or subenvironments are not; we detail those here. 

\begin{description}
\item[Case ($\multimap$ E).] Applying subtyping (\Cref{thm:sub_type}) to the induction hypothesis we have $\Delta \vdash w : \sigma$. We conclude by the ($\multimap$ E) rule.
\item[Case ($\times$ I).]
We define the $\max$ relation on any two subenvironments $\Gamma_1$ and $\Gamma_2$ so that $\max(\Gamma_1,\Gamma_2) \sqsubseteq \Gamma_1$ and $\max(\Gamma_1,\Gamma_2) \sqsubseteq \Gamma_2$. Let us denote $\max(\Gamma_1,\Gamma_2)$ by the environment $\Delta$. By the induction hypothesis and weakening (\Cref{thm:weakening}) we have  $\Delta \vdash v : \sigma $ and 
$\Delta \vdash w : \tau$. We conclude by the ($\times$ I) rule.
\item[Case ({$\tensor$ E}).] 
Let us denote by $\Theta$ the environment $\Gamma + {\max}(s_1,s_2) *
\Delta$ and by $s$ the sensitivity ${\max}(s_1,s_2)$. By the induction hypothesis and weakening (\Cref{thm:weakening}), we have \newline
${\Gamma, x:_{s}\sigma, y:_{s}\tau \vdash e : \rho}$ and we can conclude by the ($\otimes$ E) rule.
\item[Case (+ E).] 
The proof relies on the fact that, given a $\rho$ such that $\rho = \max(\rho_1,\rho_2)$, both $\rho_1$ and $\rho_2$ are subtypes of $\rho$. Using this fact, and by subtyping (\Cref{thm:sub_type}) and the induction hypothesis, we have $\Theta, x:_s \sigma \vdash e : \rho$  and $\Theta, y:_s \tau \vdash f : \rho$. If $s >0$, we can can conclude directly by the ${\text{(+ E)}}$ rule. Otherwise, we first apply weakening (\Cref{thm:weakening}) and then conclude by the ${\text{(+ E)}}$ rule.
\end{description}
\end{proof}
\fi

\subsection{Evaluation}
In order to serve as a practical tool, our type-checker must infer useful error bounds within a reasonable amount of time. 
Our empirical evaluation therefore focuses on measuring two key properties: tightness of the inferred error bounds and performance. 
To this end, our evaluation includes a comparison in terms of relative error and performance to two popular tools that {soundly} and automatically bound relative error: {FPTaylor}~\cite{FPTaylor} and {Gappa}~\cite{GAPPA}.
Although Daisy~\cite{DAISY} and Rosa ~\cite{Rosa2} also compute relative error bounds, 
they do not compute error bounds for the directed rounding modes, and our instantiation of \Lang requires round towards $+\infty$ (see \Cref{sec:examples}).
For our comparison to Gappa and FPTaylor, we use benchmarks from FPBench~\cite{FPBench}, which is the standard set of benchmarks used in the domain; we also include the Horner scheme discussed in \Cref{sec:examples}. 
There are limitations, summarized below, to the arithmetic operations that the instantiation of \Lang used in our type-checker can handle, so we are only able to evaluate a subset of the FPBench benchmarks. 
Even so, larger examples with more than 50 floating-point operations are intractable for most tools~\cite{SATIRE}, including FPTaylor and Gappa, and are not part of FPBench. On these examples, our evaluation compares against standard relative error bounds from the numerical analysis literature.
Finally, we used our type-checker to analyze the rounding error of four floating-point conditionals.

Our experiments were performed on a MacBook with a 1.4 GHz processor and 8 GB of memory. Relative error bounds are derived from the relative precision computed by \Lang using \Cref{eq:conv}.

\subsubsection{Limitations of \Lang} \label{sec:limit}
Soundness of the error bounds inferred by our type-checker 
is guaranteed by \Cref{cor:err-sound} and the instantiation of
\Lang described in \Cref{sec:examples}. 
This instantiation imposes the following limitations on the benchmarks we can consider in our evaluation. First, only the operations $+$, $*$, $/$, and {sqrt} are supported by our instantiation, so we can't use benchmarks with subtraction or transcendental functions. Second, all constants and variables must be strictly positive numbers, and the rounding mode must be fixed as round towards $+\infty$. These limitations follow from the fact that the RP metric (\Cref{def:rp}) is only well-defined for non-zero values of the same sign. We leave the exploration of tradeoffs between the choice of metric and the primitive operations that can be supported by the language to future work.
Given these limitations, along with the fact that \Lang does not currently support programs with loops, we were able to include 13 of the 129 unique (at the time of writing) benchmarks from FPBench in our evaluation.

\begin{table}
  \caption{Comparison of \Lang{} to {FPTaylor} and {Gappa}. The {Bound} column
    gives upper bounds on relative error (smaller is better); the bounds for
    {FPTaylor} and {Gappa} assume all variables are in $[0.1, 1000]$.  The
    {Ratio} column gives the ratio of \Lang's relative error bound to the
  tightest (best) bound of the other two tools; values less than 1 indicate that
\Lang{} provides a tighter bound. The {Ops} column gives the number of operations
in each benchmark.  Benchmarks from FPBench are marked with a (*).} 
\label{tab:comp} 
\begin{adjustbox}{max width=1.0\textwidth,center}
  \begin{tabular}{l c c c c c c c c} 
\hline
{Benchmark}
& {Ops}
& \multicolumn{3}{c}{Bound} 
& {Ratio}
& \multicolumn{3}{c}{Timing (s)} 
 \\ 
\cline{3-5}
\cline{7-9}
&
& {\Lang}
& {{FPTaylor}}
& {{Gappa}}
&
& {\Lang}
& {{FPTaylor}}
& {{Gappa}}
\\
\hline 
{hypot*} & 4 & 5.55e-16& {5.17e-16} & \textbf{4.46e-16} 
&1.3  & 0.002 & 3.55 & 0.069
\\ 
{x\_by\_xy*}& 3  &{4.44e-16} & fail & \textbf{2.22e-16} 
& 2 & 0.002 & -  & 0.034 
\\
{one\_by\_sqrtxx} & 3 &{5.55e-16} & 5.09e-13 & \textbf{3.33e-16} 
& 1.7 & 0.002 & 3.34  & 0.047 
\\
{{sqrt\_add*} }  & 5 & 9.99e-16 & {6.66e-16} & \textbf{5.54e-16} 
& 1.5  & 0.003 & 3.28 & 0.092
\\
{test02\_sum8*} & 8 & {\textbf{1.55e-15}} & 9.32e-14 & {\textbf{1.55e-15}} 
& 1 &  0.002 & 14.61 & 0.244 
\\
{nonlin1*}  & 2 & {4.44e-16} & {4.49e-16} & \textbf{2.22e-16} 
& 2 & 0.003 & 3.24 & 0.040 
\\
{test05\_nonlin1*} & 2 & {4.44e-16} & {4.46e-16} & \textbf{2.22e-16} 
& 2 & 0.008 & 3.27 & 0.042 
\\
{verhulst*} & 4 & {8.88e-16} & {7.38e-16} & \textbf{{4.44e-16}} 
& 2 & 0.002 & 3.25 & 0.069 
\\
{predatorPrey*} & 7 & {1.55e-15} & {4.21e-11} & \textbf{8.88e-16}
& 1.7 & 0.002 & 3.28 & 0.114 
\\
{test06\_sums4\_sum1*} & 4 & {\textbf{6.66e-16}} & {6.71e-16} & {\textbf{6.66e-16}}
& 1 & 0.003 & 3.84 & 0.069
\\
{test06\_sums4\_sum2*} & 4 & {6.66e-16} & {1.78e-14} & 
\textbf{4.44e-16} & 1.5 & 0.002 & 11.02 & 0.055 
\\
{i4*} & 4 & {\textbf{4.44e-16}} & {4.50e-16} & {\textbf{4.44e-16}} 
& 1 & 0.002 & 3.30 & 0.055 
\\
{Horner2} & 4 & {\textbf{4.44e-16}} & 6.49e-11 & {\textbf{4.44e-16}}
& 1  & 0.002 & 11.72  & 0.052 
\\
{{Horner2\_with\_error}} & 4 & {1.55e-15} & 1.61e-10 & 
\textbf{1.11e-15} & 1.4 & 0.002 & 19.56 & 0.119
\\
{Horner5} & 10 & {\textbf{1.11e-15}} & 1.62e-01 & {\textbf{1.11e-15}}
& 1 & 0.003 & 22.08  & 0.209
\\
{Horner10} & 20 & {\textbf{2.22e-15}} & 1.14e+13 & {\textbf{2.22e-15}}
& 1 & 0.003 & 40.68  & 0.650
\\
{Horner20} & 40  & {\textbf{4.44e-15}} & 2.53e+43 & \textbf{{4.44e-15}}
& 1 & 0.003 & 109.42  & 2.246
\\
\hline 
\end{tabular}
\end{adjustbox}
\end{table}

\subsubsection{Small Benchmarks}
The results for benchmarks with fewer than 50 floating-point operations are given in \Cref{tab:comp}. Eleven of the seventeen benchmarks are taken from the FPBench benchmarks.
Both {FPTaylor} and {Gappa} require 
user provided interval bounds on the input variables in order to compute the relative error; we used an
interval of $[0.1,1000]$ for each of the benchmarks. 
We used the default configuration for {FPTaylor}, and used 
{Gappa} without providing hints for interval subdivision. 
The floating-point format of each benchmark is binary64, and the rounding mode is set at round towards $+\infty$; the unit roundoff in this setting is $2^{-52}$ (approximately {$2.22\text{e-}16$}).
Only \lstinline{Horner2_with_error} assumes error in the inputs.

\subsubsection{Large Benchmarks}
\Cref{tab:large} shows the results for benchmarks with 100 or more floating-point operations. Five of the nine benchmarks are taken 
from \textsc{Satire}~\cite{SATIRE}, an \emph{empirically sound} static analysis tool that computes absolute error bounds. 
Although \textsc{Satire} does not statically compute relative error bounds for the benchmarks listed in \Cref{tab:large}, most of these benchmarks have well-known worst case relative error bounds that we can compare against. 
These bounds are given in the {Std.} column in \Cref{tab:large}; the relevant references are as follows: Horner's scheme ~\citep[cf.][p. 95]{HighamBook}, summation~\citep[cf.][p. 260]{boldo_jeannerod_melquiond_muller_2023}, and matrix multiply~\citep[cf.][p. 63]{HighamBook}. For matrix multiplication, we report the max element-wise relative error bound produced by \Lang.
When available, the Timing column in \Cref{tab:large} lists the time reported for \textsc{Satire} to compute \emph{absolute} error bounds~\citep[cf.][Table III]{SATIRE}. 

\subsubsection{Conditional Benchmarks}
\Cref{tab:cond} shows the results for conditional benchmarks. Two of the four benchmarks are taken from FPBench and the remaining benchmarks are examples from Dahlquist and Bj\"{o}rck ~\citep[cf.][p. 119]{Dahlquist}.
We were unable to compare the performance and computed relative error bounds shown in \Cref{tab:cond} against other tools. While Daisy, FPTaylor, and Gappa compute relative error bounds, they don't handle conditionals. And, while PRECiSA can handle conditionals, it doesn't compute relative error bounds. Only Rosa computes relative error bounds for floating-point conditionals, but Rosa doesn't compute bounds for the directed rounding modes. 

\begin{table}
  \caption{
The performance of \Lang on benchmarks with 100 or more floating-point operations.  
The {Std.} column gives relative error bounds from the literature.
Benchmarks from  \textsc{Satire} are marked with with an ({a});
the \textsc{Satire} subcolumn gives timings for the computation of \emph{absolute} error bounds as reported in ~\citep{SATIRE}.
} 
\label{tab:large} 
  \begin{tabular}{l c c c c c} 
\hline
{Benchmark}
& {Ops}
& {Bound (\Lang)} 
& {Bound (Std.)} 
& \multicolumn{2}{c}{Timing (s)} 
\\
\cline{5-6}
&
&
&
& {\Lang}
& {{\textsc{Satire}}}

\\
\hline 
\small{Horner50}$^{a}$& 100  &{1.11e-14}  & {1.11e-14} & 9e-03 & 5
\\
\small{MatrixMultiply4}& 112  &{1.55e-15}  & {8.88e-16} & 3e-03 & -
\\
\small{Horner75}& 150  &{1.66e-14} & {1.66e-14} & 2e-02 & -
\\
\small{Horner100}& 200  &{2.22e-14} & {2.22e-14} & 4e-02 & -
\\
\small{SerialSum}$^{\text{a}}$& 1023  &{2.27e-13}  &  {2.27e-13}  & 5 & 5407
\\
\small{{Poly50}}$^{\text{a}}$  & 1325 &{2.94e-13} &  - & 2.12 & 3
\\
\small{MatrixMultiply16}& 7936  &{6.88e-15}  & {3.55e-15} & 4e-02 & -
\\
\small{MatrixMultiply64}$^{\text{a}}$ & 520192 &{2.82e-14}  & {1.42e-14} & 10 & 65
\\
\small{MatrixMultiply128}$^{\text{a}}$ & 4177920  &{5.66e-14}  & {2.84e-14} & 1080 & 763
\\
\hline 
\end{tabular}
\end{table} 
\begin{table}
  \caption{
The performance of \Lang on conditional benchmarks. Benchmarks 
from FPBench are marked with with a (${*}$). 
Benchmarks 
from Dahlquist and Bj\"{o}rck ~\citep[cf.][p. 119]{Dahlquist} are marked with with a ($\text{b}$). 
} \label{tab:cond} 
  \begin{tabular}{l c c } 
\hline
{Benchmark}
& {Bound} 
& {Timing (ms)} 

\\
\hline 
\small{PythagoreanSum}$^{\text{b}}$ &{8.88e-16} & 2 
\\
\small{HammarlingDistance}$^{\text{b}}$  &{1.11e-15} & 2 
\\
\small{squareRoot3}$^{*}$  &{4.44e-16} & 2 
\\
 \small{{squareRoot3Invalid}}$^{*}$  &{4.44e-16} & 2
\\
\hline 
\end{tabular}
\end{table} 

\subsubsection{Evaluation Summary}

We draw three main conclusions from our evaluation. 

\textbf{ \emph{Roundoff error analysis via type checking is fast.}}
On small and conditional benchmarks, \Lang infers an error bound in the order of milliseconds. This is at least an order of magnitude faster than either {Gappa} or {FPTaylor}. On larger benchmarks, \Lang's performance surpasses that of comparable tools by computing bounds for problems with up to 520k operations in under a minute.

\textbf{ \emph{Roundoff error bounds derived via type checking are useful.}}
On most small benchmarks \Lang produces a relative error bound that is tighter than FPTaylor, and either matching or within a factor of two of Gappa. For benchmarks where rounding errors are composed and magnified, such as \lstinline{Horner2_with_error}, and on somewhat larger benchmarks like \lstinline{Horner2}-\lstinline{Horner20}, our type-based approach performs particularly well. On larger benchmarks that are intractable for the other tools, \Lang produces bounds that nearly match the known worst-case bounds from the literature. \Lang is also able to provide non-trivial relative error bounds for floating-point conditionals.

\textbf{ \emph{Roundoff error bounds derived via type checking are strong.}} 
The relative error bounds produced by \Lang
hold for all positive real inputs, assuming the absence of overflow and underflow. In comparison, the relative error bounds derived by  
{FPTaylor} and {Gappa} only hold for the
user provided interval bounds on the input variables, which we took to be $[0.1,1000]$. Increasing this interval range allows 
{FPTaylor} and {Gappa} to give stronger bounds, but can also lead to
slower analysis.
Furthermore, given that relative error is poorly behaved for values near zero, some tools are sensitive to the choice of interval.
We see this in the results for the benchmark \lstinline{x_by_xy}
in \Cref{tab:comp}, where we are tasked with calculating the roundoff error produced by the expression $x/(x+y)$, where $x$ and $y$ are binary64
floating-point numbers in the interval $[0.1,1000]$. 
For these parameters, the expression lies in the interval  
$[5.0\text{e-}05,1.0]$ and the relative error should still 
be well defined. However, {FPTaylor} (used with its default configuration) fails to provide a bound, and issues a warning due to the 
value of the expression being too close to zero. 
\begin{remark}[User specified Input Ranges]
Allowing users to specify input ranges is a feature of many tools used for floating-point error analysis, including FPTaylor and Gappa. 
In some cases, a useful bound can't be computed for an unbounded range, 
but can be computed given a well-chosen bounded range for the inputs. 
Input ranges are also required for computing absolute error bounds.
Extending \Lang with bounded range inputs is left to future work; we note that this feature could be supported by adding a new type to the language, and by adjusting the types of primitive operations.
    
\end{remark}

\section{Extending the Neighborhood Monad}\label{sec:extensions}

So far, we have seen how the graded neighborhood monad can model rounding error when the rounding operation follows a standard rounding rule (round towards $+\infty$) and assuming that underflow and overflow do not occur. In this section, we propose variations of this monad to support error analysis for rounding operations with more complex behavior.

\subsection{Extension: Non-Normal Numbers}\label{subsec:non-normal}

In practice, rounding the result of a floating-point operation might
result in a \emph{non-normal} value: 
numbers that are not too
small (\emph{underflow}) or too large (\emph{overflow}) for the size of
floating-point representation. 
For a more realistic model of rounding, 
we can adjust the semantics of our language to
accurately model non-normal values.

\paragraph{Extending the graded monad.}

First, we extend the neighborhood monad with exceptional values.

\begin{definition} \label{def:nhd-monad-exc}
  Let the pre-ordered monoid $\mathcal{R}$ be the extended non-negative real
  numbers $\NNR \cup \{ \infty \}$ with the usual order and addition, and let
  $\diamond$ be a special element representing an exceptional value. The
  \emph{exceptional neighborhood monad} is defined by the functors $\{ T_r^* : \Met \to \Met
  \mid r \in \mathcal{R} \}$:
  \begin{itemize}
    \item $T_r^* : \Met \to \Met$ maps a metric space
      $M$ to the metric space with underlying set
      \[
        |T_r^* M| \triangleq \{ (x, y) \in M \times (M \cup \{ \diamond \})
        \mid d_M(x, y) \leq r \text{ or } y = \diamond \}
      \]
      and metric
      \[
        \begin{cases}
          d_{T_r^* M} ( (x, y), (x', y') ) &\triangleq d_M (x, x') \\
          d_{T_r^* M} ( (x, y), \diamond ) &\triangleq 0
        \end{cases}
      \]
    \item $T_r^*$ takes a non-expansive function $f : A \to B$ to a
      function $T_r^* f : T_r^* A \to T_r^* B$ defined via:
      \[
        (T_r^* f)( (x, y) ) \triangleq \begin{cases}
          (f(x), f(y)) &: y \in A \\
          (f(x), \diamond) &: y = \diamond
        \end{cases}
      \]
  \end{itemize}
  \ifshort
\else
  The associated natural transformations are defined as follows:
  \begin{itemize}
    \item For $r, q \in \mathcal{R}$ and $q \leq r$, the map $(q \leq r)_A :
      T_q^* A \to T_r^* A$ is the identity.
    \item The unit map $\eta_A : A \to T_0^* A$ is defined via:
      \[
        \eta_A(x) \triangleq (x, x)
      \]
    \item The graded multiplication map $\mu_{q, r, A} : T_q^* (T_r^* A) \to T_{q + r}^*
      A$ is defined via:
      \[
        \begin{cases}
          \mu_{q, r, A} ( (x, y), (x', y') ) &\triangleq (x, y') \\
          \mu_{q, r, A} ( (x, y), \diamond ) &\triangleq (x, \diamond)
        \end{cases}
      \]
  \end{itemize}
\fi
\end{definition}
$T_r^*$ are evidently functors, and the associated maps are natural
transformations. 
%
\ifshort\else
\begin{definition} \label{def:nhd-st-exc}
  Let $r \in \mathcal{R}$. We define a family of non-expansive maps:
  \begin{align*}
    st_{r, A, B}^* &: A \otimes T_r^* B \to T_r^* (A \otimes B) \\
    st_{r, A, B}^*(a, (b, b')) &\triangleq \begin{cases}
        ( (a, b), (a, b') ) &: b' \neq \diamond \\
        ( (a, b), \diamond ) &: b' = \diamond \\
      \end{cases}
  \end{align*}
  making $T_r^*$ a strong monad. Furthermore, for any $s \in \mathcal{S}$ the
  identity map is a non-expansive map with the following type:
  \[
    \lambda_{s, r, A}^* : D_s (T_r^* A) \to T_{s \cdot r}^* (D_s A)
  \]
\end{definition}
\fi

\paragraph{Extending the $\Met$ semantics.}

We can define the exceptional metric semantics $\denot{\Gamma \vdash e : \tau}^*
: \denot{\Gamma} \to \denot{\tau}$ as before (\Cref{def:interp-prog}), using the monad
$T_r^*$ instead of $T_r$. The only change is that rounding operations can now produce exceptional values. We assume that rounding is interpreted by a function $\rho^* : R \to (R \cup \{ \diamond \})$ where $\diamond$ represents any exceptional value (e.g., underflow or overflow). We continue to assume that the numeric type is interpreted by a metric space $\denot{\num} = (R, d_R)$. For all numbers $r \in R$, we require that $\rho^*$ satisfy
$d_{R}(r, \rho^*(r)) \leq \rnderr$
whenever $\rho^*(r)$ is not the value $\diamond$.

Letting $f = \denot{\Gamma \vdash k : \num}^*$, we then define
$
  \denot{\Gamma \vdash \mathbf{rnd}~k : M_{\rnderr} \num}^*
  \triangleq f ; \langle id, \rho^* \rangle
$

\paragraph{Extending the FP semantics.}

To account for the floating point operational semantics possibly producing
exception values, we introduce a new error value with a new typing rule:
\[
  v, w ::= \cdots \mid \mathbf{err}
  \qquad\qquad\qquad\qquad
  \inferrule[Err]
  { }
  { \Gamma \vdash \mathbf{err} : M_u \tau }
\]
We only consider this value for the floating-point semantics---programs under
the metric and real semantics cannot use $\mathbf{err}$, and never step to
$\mathbf{err}$.

To interpret the monadic type in the floating-point semantics, we use the Maybe
monad:
\[
  \pdenot{M_u \tau}_{fp}^* \triangleq \pdenot{\tau}_{fp}^* \uplus \{ \diamond \}
\]
The floating-point semantics remains the same as before (\Cref{def:id-fp-steps})
except for two changes. First, we interpret the rule \textsc{Err} by letting
$\pdenot{\Gamma \vdash \mathbf{err} : M_u \tau}_{fp}^*$ be the constant function
producing $\diamond$. Second, given $f = \pdenot{\Gamma \vdash k :
\num}_{fp}^*$, we define $\pdenot{\Gamma \vdash \mathbf{rnd}~k : M_{\rnderr}
\num}_{fp}^* \triangleq f ; \rho^*$. Note that the function $\rho^*$ may produce
the exceptional value $\diamond$.

On the operational side, we modify the evaluation rule for round:
\[
  \mathbf{rnd}~k \mapsto_{fp} \begin{cases}
    r &: \rho^*(k) = r \in R \\
    \mathbf{err} &: \rho^*(k) = \diamond
  \end{cases}
\]
And add a new step rule for propagating exceptional values:
$  \letbind(\mathbf{err}, x.f) \mapsto_{fp} \mathbf{err}$.

\paragraph{Establishing error soundness.}
\ifshort
\else
It is straightforward to show that these step rules are sound for the
floating point semantics; the proof is the analogue of \Cref{lem:pres-id-fp}.

\begin{lemma}[Preservation] \label{lem:pres-id-fp-exc}
  Let $\cdot \vdash e : \tau$ be a well-typed closed term, and suppose $e
  \mapsto_{fp} e'$. Then there is a derivation of $\cdot \vdash e' : \tau$ such
  that $\pdenot{\vdash e : \tau}^*_{fp} = \pdenot{\vdash e' : \tau}^*_{fp}$.
\end{lemma}

Finally, we have the same paired soundness theorem as before:

\begin{lemma}[Pairing]\label{lem:pairing-exc}
  Let $\cdot \vdash e : M_r \num$. Then we have:
  \[
    U \denot{e}^* = \langle \pdenot{e}_{id}, \pdenot{e}_{fp}^* \rangle
  \]
  in $\Set$: the first projection of $U \denot{e}$ is $\pdenot{e}_{id}$, and the
  second projection is $\pdenot{e}_{fp}^*$.
\end{lemma}

The proof is essentially identical to \Cref{lem:pairing}. Finally, we have the
following analogue of \Cref{cor:err-sound}.
\fi

\begin{corollary} \label{cor:err-sound-exc}
  Let $\cdot \vdash e : M_r \num$ be well-typed. Under the exceptional
  semantics, either:
    $e \mapsto_{id} \mathbf{ret}~v_{id}$ and $e \mapsto_{fp} \mathbf{ret}~v_{fp}$,
      and $d_{\denot{\num}}(\denot{v_{id}}^*, \denot{v_{fp}^*}) \leq r$, or
     $e \mapsto_{fp} \mathbf{err}$.
\end{corollary}

Thus, the error bound holds assuming floating point evaluation does not hit an
exceptional value.

\subsection{Further Extension}\label{subsec:other-ext}

The exceptional neighborhood monad can be viewed as the composition of two
monads: the neighborhood monad on $\Met$ models distance bounds, while the Maybe
monad on $\Set$ models exceptional behavior. By replacing the Maybe monad with
monads for other effects, we can define variants of the neighborhood monad
modeling non-deterministic and probabilistic rounding.


\ifshort
\else
\paragraph{Non-deterministic rounding.}
In effectful programming languages, non-deterministic choice can be modeled by
the \emph{powerset} monad $P$, which maps any set $X$ to the powerset $2^X = \{
A \mid A \subseteq X \}$. This monad comes with a family of maps $\tilde{\mu}_X
: P(PX) \to PX$, which simply take the union, and a family of maps
$\tilde{\eta}_X : X \to PX$ which returns the singleton set.

We can define variants of the neighborhood monad based on the powerset monad.
For any $r \in \mathcal{R}$, we can define functors $TP_r : \Met \to \Met$ with
carrier sets
\begin{align*}
  TP_r^+(X, d_X) 
  &\triangleq \{ (x, A) \in X \times PX \mid \text{for all } a \in A,\ d_X(x, a) \leq r \} \\
  TP_r^-(X, d_X)
  &\triangleq \{ (x, A) \in X \times PX \mid \text{exists } a \in A,\ d_X(x, a) \leq r \}
\end{align*}
and both metrics based on the first component: $d( (x, A), (y, B) ) \triangleq
d_X (x, y)$. The associated maps for this monad are inherited from the
associated maps of $P$. For both $TP_r^+$ and $TP_r^-$, we define:
\begin{align*}
  \eta_X(x) &= (x, \tilde{\eta}_X(x)) \\
  \mu_X((x, A), \{ (y_i, B_i) \mid i \in I \}) &= (x, \tilde{\mu}_X (\{ B_i \mid i \in I \}))
\end{align*}
\begin{theorem}
  The maps $\eta_X$ and $\mu_X$ make $TP_r^+$ and $TP_r^-$ into a
  $\mathcal{R}$-graded monad.
\end{theorem}
\begin{proof}
  It is clear that $\eta_X : X \to TP_0 X$ on the carrier sets; this map is
  non-expansive by the definition of the metric on $TP_r$. We can also check
  that $\mu_X : TP_q (TP_r X) \to TP_{q + r} X$ for both variants; we consider
  $TP_r^+$, but the other variant is exactly the same. Suppose that
  \[
    ((x, A), \{ (y_i, B_i) \mid i \in I \}) \in TP_q^+ (TP_r^+ X) .
  \]
  Then by definition of $TP_r^+$, the distance $d_X$ between $y_i$ and every
  element in $B_i$ is at most $r$, and by definition of $TP_q^+$, the distance
  $d_{TP_r^+ X}$ between $(x, A)$ and every element $(y_i, B_i)$ is at most $r$,
  so the $d_X$ distance between $x$ and $y_i$ is at most $r$ for every $i$ by
  definition of the metric. Thus by the triangle inequality, the distance $d_X$
  between $X$ and every element of $B_i$ is at most $q + r$ as desired.
  Non-expansiveness of this map is straightforward, the two diagrams required
  for a graded monad follow from the monad laws for $P$.
\end{proof}

This monad can model numerical computations with non-deterministic aspects. For
instance, rounding could behave non-deterministically in the event of ties, or
when the standard allows implementation-specific behavior.

The two variants of the neighborhood monad provide different guarantees: using
$TP_r^+$ ensures that \emph{all} resolutions of the non-determinism differ from
the ideal value by at most $r$, while using $TP_r^-$ ensures that \emph{some}
resolution of the non-determinism differs from the ideal value by at most $r$;
this situation seems analogous to \emph{may} versus \emph{must}
non-determinism~\citep{lassen-thesis}.

\paragraph{State-dependent rounding.}
Another typical monad for computational effects is the \emph{global state} monad
$S$. Given a fixed set $\Sigma$ of possible global states, the state monad maps
any set $X$ to the set of functions $\Sigma \to \Sigma \times X$. The state
monad comes with a family of maps $\tilde{\mu}_X : S(SX) \to SX$, which flattens
a nested stateful computation by sequencing, and a family of maps
$\tilde{\eta}_X : X \to SX$, which maps $x \in X$ to the stateful computation
that preserves its state and always produces $x$.

For any $r \in \mathcal{R}$, we can define functors $TS_r : \Met \to \Met$ with
carrier sets
\begin{align*}
  TS_r(X, d_X) 
  &\triangleq \{ (x, f) \in X \times SX \mid \text{for all } s \in S,\ f(s) = (s', a)
  \text{ and } d_X(x, a) \leq r \}
\end{align*}
and metrics based on the first component: $d( (x, f), (y, g) ) \triangleq d_X
(x, y)$. The associated maps for this monad are inherited from the associated
maps of $S$:
\begin{align*}
  \eta_X(x) &= (x, \tilde{\eta}_X(x)) \\
  \mu_X( (x, f), g \in S(TS_r X) ) &= (x, S(\pi_2)(g) ; \tilde{\mu}_X)
\end{align*}

\begin{theorem}
  The maps $\eta_X$ and $\mu_X$ make $TS_r$ into a $\mathcal{R}$-graded monad.
\end{theorem}
\begin{proof}
  It is clear that $\eta_X : X \to TS_0 X$ on the carrier sets; this map is
  non-expansive by the definition of the metric on $TS_r$. We can also check
  that $\mu_X : TS_q (TS_r X) \to TS_{q + r} X$. Suppose that
  \[
    ((x, f \in SX), g \in S(TS_r X)) \in TS_q (TS_r X) .
  \]
  Then for every state $\sigma \in \Sigma$, $g(\sigma) = (\sigma', (y, h)) \in
  \Sigma \times TS_r X$ so for every state $\sigma' \in \Sigma$, the distance
  $d_X$ between $y$ and the output $z$ in $(-, z) = h(\sigma')$ is at most $r$.
  By the definition of $TS_q$, the distance between $(x, f)$ and $g(\sigma) =
  (y, h)$ is at most $q$, so the distance $d_X$ between $x$ and $y$ is at most
  $q$. Thus by the triangle inequality, the distance between $x$ and $z$ is at
  most $q + r$ as desired. Non-expansiveness of this map is straightforward,
  the two diagrams required for a graded monad follow from the monad laws for
  $S$.
\end{proof}
This monad can be used to model rounding behavior that depends on the
machine-specific state, such as floating point registers. The definition of
$TS_r$ ensures that stateful computations always produce values that are at most
$r$ from the ideal value, regardless of the initial state of the system.

\paragraph{Randomized rounding.}
Finally, we can consider layering the neighborhood monad with monads for
randomized choice. For instance, the \emph{(finite) distribution} monad $D$ maps
any set $X$ to the set of probability distributions over $X$ with finite
support. This monad comes with a family of maps $\tilde{\mu}_X : D(DX) \to DX$,
which flattens a distribution over distributions, and a family of maps
$\tilde{\eta}_X : X \to DX$, which returns a point (deterministic) distribution.
Much like our previous monads, we can define functors $TD_r : \Met \to \Met$
where the carrier set of $TD_r X$ is a subset of $X \times DX$. There are at
least three choices for which subset to take, however:
\begin{enumerate}
  \item $(x, p : DX) \in TD_rX$ iff $d_X(x, a) \leq r$ for \emph{all} $a$ in the
    support of $p$.
  \item $(x, p : DX) \in TD_rX$ iff $d_X(x, a) \leq r$ for \emph{some} $a$ in
    the support of $p$.
  \item $(x, p : DX) \in TD_rX$ iff the \emph{expected value} of $d_X(x, a)$
    when $a$ is drawn from $p$ is at most $r$.
\end{enumerate}
The unit $\eta_X$ and multiplication $\mu_X$ maps are inherited from
$\tilde{\eta}_X$ and $\tilde{\mu}_X$, just as before.

\begin{theorem}
  The maps $\eta_X$ and $\mu_X$ make all variants of $TD_r$ into a
  $\mathcal{R}$-graded monad.
\end{theorem}
\begin{proof}
  The first two variants follow from the essentially the same argument for
  $TP_r^+$ and $TP_r^-$, respectively. We focus on the third variant. As with
  the other cases, the most intricate point to check is that $\mu_X : TD_q (TD_r
  X) \to TD_{q + r} X$ on carrier sets. The argument follows from the triangle
  inequality and the law of total expectations. We write $[ x_i : \alpha_i \mid
  i \in I ]$ where $x_i \in X$ and $\alpha_i \in [0, 1]$ sum up to $1$ for a
  distribution in $DX$. Suppose that:
  \[
    ((x, -), p = [ (y_i, p_i) : \alpha_i \mid i \in I ]) \in TD_q (TD_r X)
    \qquad\text{and}\qquad
    p_i = [ z_i : \beta_{ij} \mid j \in J ] \in DX
  \]
  Then by definition of $TD_r X$, we have $\mathbb{E}_{z \sim p_i} [ d_X(y_i, z)
  ] \leq r$ for every $i$, and by definition of $TD_q X$ and the metric on
  $TD_q$, we have $\mathbb{E}_{y \sim p} [ d_X(x, y) ] \leq q$. By definition,
  $\mu_X$ maps the input $((x, -), p)$ to $(x, m)$, where $m = D(\pi_2)(p) ;
  \tilde{\mu}_X$. Now we can compute:
  \begin{align}
    \mathbb{E}_{z \sim m} [ d_X(x, z) ]
    &= \mathbb{E}_{y_i \sim p} [ \mathbb{E}_{z \sim p_i} [ d_X(x, z) ] ]
    \tag{def. $\tilde{\mu}_X$} \\
    &\leq \mathbb{E}_{y_i \sim p} [ \mathbb{E}_{z \sim p_i} [ d_X(x, y) + d_X(y, z)] ]
    \tag{triangle ineq.} \\
    &= \mathbb{E}_{y_i \sim p} [ d_X(x, y) + \mathbb{E}_{z \sim p_i} [ d_X(y, z) ] ]
    \tag{constant} \\
    &\leq \mathbb{E}_{y_i \sim p} [ d_X(x, y) + r ]
    \leq q + r
    \tag{def. $TD_r$ and $TD_q$}
  \end{align}
\end{proof}

The first two choices provide worst-case and best-case bounds on numerical error
when rounding may be probabilistic. The last choice is particularly interesting:
it provides an \emph{average-case} bound on the numerical error between the
floating-point value and the ideal value.

\paragraph{Towards a more uniform picture.}
We have outlined several examples of composing the graded neighborhood monad
with monads for computational effects, but we do not understand the picture in 
full generality.  We see several interesting questions, which we leave for
future work:
\begin{itemize}
  \item How can we define the carrier sets for the graded monad? For some of our
    computational monads we seem to be able to define multiple versions of the
    graded monad, but currently it is not clear which effectful monads support
    which graded versions.
  \item What properties of the effect monad carry over to the neighborhood
    monad? For instance, we conjecture that as long as the effect monad is
    strong, then the graded monad is also strong. It is currently not clear what
    other properties of the effect monad are preserved.
  \item How generally does this construction work? The main ingredients in our
    construction seem to be the neighborhood monad on $\Met$ and an effect
    monad over $\Set$. A natural question is whether these categories can be
    generalized further.
\end{itemize}
\fi

\section{Related Work}\label{sec:relatedwork}

\paragraph{Type systems for floating-point error.}
A type system due to \citet{DBLP:conf/icfem/Martel18} uses dependent types to
track roundoff errors. A significant difference between \Lang and the type
system proposed by Martel is error soundness. In Martel's system, the soundness
result says that a semantic relation capturing the notion of accuracy between a
floating-point expression and its ideal counterpart is preserved by a reduction
relation. This is weaker than a standard type soundness guarantee. In
particular, it is not shown that  well-typed terms satisfy the semantic
relation.  In \Lang, the central novel property guaranteed by our type system
is much stronger: well-typed programs of monadic type satisfy the error bound
indicated by their type.

\paragraph{Program analysis for roundoff error.}
Many verification methods have been developed to automatically bound roundoff
error. The earliest tools, like {Fluctuat}~\citep{Fluctuat} and
{Gappa}~\citep{GAPPA}, employ abstract interpretation with interval arithmetic
or polyhedra~\citep{DBLP:conf/aplas/ChenMC08} to overapproximate the range of
roundoff errors. This method is flexible and applies to general programs with
conditionals and loops, but it can significantly overestimate roundoff error,
and it is difficult to model cancellation of errors.

To provide more precise bounds, recent work relies on optimization.
Conceptually, these methods bound the roundoff error by representing the error
symbolically as a function of the program inputs and the error variables
introduced during the computation, and then perform global optimization over
all settings of the errors.  Since the error expressions are typically complex,
verification methods use approximations to simplify the error expression to
make optimization more tractable, and mostly focus on straight-line programs.
For instance, {Real2Float}~\citep{REAL2FLOAT} separates the error expression
into a linear term and a non-linear term; the linear term is bounded using
semidefinite programming, while the non-linear term is bounded using interval
arithmetic. FPTaylor~\cite{FPTaylor} was the first tool to use Taylor
approximations of error expressions. 
\citet{10.1007/978-3-031-44245-2_4} describe a modular method for bounding the
propagation of errors using Taylor approximations, and
{Rosa}~\citep{Rosa1,Rosa2} uses Taylor series to approximate the propagation of
errors in possibly non-linear terms.

In contrast, our type system does not rely on global optimization, can
naturally accommodate both relative and absolute error, and can be instantiated
to different models of floating-point arithmetic with minimal changes. Our
language supports a variety of datatypes and higher-order functions. While our
language does not support recursive types and general recursion, similar
languages support these features~\citep{Fuzz,DBLP:conf/popl/AmorimGHKC17,DLG}
and we expect they should be possible in \Lang; however, the precision of the
error bounds for programs using general recursion might be poor. Another
limitation of our method is in typing conditionals: while \Lang can only derive
error bounds when the ideal and floating-point executions follow the same
branch, tools that use general-purpose solvers (e.g., {PRECiSA} and {Rosa}) can
produce error bounds for programs where the ideal and floating-point executions
take different branches.

\paragraph{Verification and synthesis for numerical computations.} 
Formal verification has a long history in the area of numerical computations,
starting with the pioneering work of Harrison
~\cite{10.1007/3-540-48256-3_9,10.1007/BFb0000475,10.1007/3-540-40922-X_14}.
Formalized specifications of floating-point arithmetic have been developed in
the Coq~\cite{Flocq}, Isabelle~\cite{IEEE_Floating_Point-AFP}, and
PVS~\cite{MINER199531} proof assistants. These specifications have been used to
develop sound tools for floating-point error analysis that generate proof
certificates, such as {VCFloat}~\cite{VCFloat1,Appel:CPP:2024} and
{PRECiSA}~\cite{PRECISA}.  They have also been used to  mechanize proofs of
error bounds for specific numerical programs~(e.g.,
\cite{DBLP:conf/cav/KellisonA22,boldo2014,10.1007/978-3-031-42753-4_14,Kellison:Arith:2023,Moscato:2019}).
Work by \citet{ASTREE} has applied abstract interpretation to verify the
absence of floating-point faults in flight-control software, which have caused
real-world accidents.  Finally, recent work uses \emph{program synthesis}:
{Herbie}~\citep{Herbie} automatically rewrites numerical programs to reduce
numerical error, while {RLibm}~\citep{RLIBM} automatically generates
correctly-rounded math libraries.

\paragraph{Type systems for sensitivity analysis.}

\Lang belongs to a line of work on linear type systems for sensitivity
analysis, starting with Fuzz~\citep{Fuzz}. We point out a few especially
relevant works. Our syntax and typing rules are inspired by~\citet{DLG}, who
propose a family of Fuzz-like languages and define various notions of
operational equivalence; we are inspired by their syntax, but our case
elimination rule ($+$ E) is different: we require $s$ to be strictly positive
when scaling the conclusion. This change is due to a subtle difference in how
sums are treated. 
\ifshort
\else
In \Lang, as in Fuzz, the distance between left and right injections is
$\infty$, whereas in the system by \citet{DLG}, left and right injections are
not related at any distance. Our approach allows non-trivial operations
returning booleans to be typed as infinite sensitivity functions, but the case
rule must be adjusted: to preserve soundness, the
conclusion must retain a dependence on the guard expression, even if the guard
is not used in the branches.
\fi
\citet{Amorim2019} added a graded monadic type to Fuzz to handle more complex
variants of differential privacy; in their application, the grade does not
interact with the sensitivity language. 

\paragraph{Other approaches to error analysis.}

The numerical analysis literature has explored other conceptual tools for
static error analysis, such as stochastic error
analysis~\citep{doi:10.1137/20M1334796}.  Techniques for \emph{dynamic} error
analysis, which estimate the rounding error at runtime, have also been
proposed~\citep{HighamBook}.  \ifshort \else It would be interesting to
consider these techniques from a formal methods perspective, whether by
connecting dynamic error analysis with ideas like runtime verification, or
developing methods to verify the correctness of dynamic error analysis.  \fi

\section{Conclusion and Future Directions}\label{sec:conclusion}

We have presented a type system for bounding roundoff error, combining two
elements: a sensitivity analysis through a linear type system, and error
tracking through a novel graded monad. Our work demonstrates that type systems
can reason about quantitative roundoff error. There is a long history of
research in error analysis, and we believe that we are just scratching the
surface of what is possible with type-based approaches.

We briefly comment on two promising directions.
First, numerical analysts have studied probabilistic models of
roundoff errors, which can give better bounds on error
in practice~\cite{doi:10.1137/18M1226312}. Combining
our system with a probabilistic language might enable a similar analysis.
Second, the error bounds we establish are \emph{forward error} bounds, 
because they bound the error in the output. In practice,
numerical analysts often consider \emph{backward error} bounds, which describe how
much the \emph{input} needs to be perturbed in order to realize the approximate
output. Such bounds can help clarify whether the source of the error is due to
the computation, or inherent in the problem instance. Tackling this kind of
property is an interesting direction for future work.

\section*{Artifact}
The artifact for the implementation of \Lang described in \Cref{sec:implementation} is available online \cite{NumFuzz:artifact:2024}. It includes instructions on how to reproduce the results reported in \Cref{tab:comp,tab:large,tab:cond}.

\begin{acks}
  We thank Pedro Henrique Azevedo de Amorim, David Bindel, Eva Darulova, Max
  Fan, and the anonymous reviewers for their close reading and useful
  suggestions. We also thank Guillaume Melquiond for suggesting some
  improvements to our Gappa scripts. Preliminary versions of this work were
  presented at Cornell's PLDG seminar and the NJ Programming Languages and
  Systems Seminar. This material is based upon work supported by the U.S.
  Department of Energy, Office of Science, Office of Advanced Scientific
  Computing Research, Department of Energy Computational Science Graduate
  Fellowship under Award Number DE-SC0021110. This work was also partially
  supported by the NSF (\#1943130) and the ONR (\#N00014-23-1-2134).
\end{acks}

\bibliographystyle{ACM-Reference-Format}
\bibliography{header,bib}

\appendix

\end{document}